\documentclass{article}
\usepackage{graphicx}

\usepackage[asymmetric,margin=1.375in, left=1.5in, right=1.5in, marginparwidth=1.3in]{geometry}

\usepackage{bbm}

\usepackage{amsmath,amsthm,amssymb}
\usepackage{bm}
\usepackage{thmtools}
\usepackage{thm-restate}

\newtheorem{lemma}{Lemma}
\newtheorem{theorem}{Theorem}

\newtheorem{problem}{Problem}
\newtheorem{proposition}{Proposition}
\newtheorem{fact}{Fact}

\theoremstyle{definition}
\newtheorem{remark}{Remark}

\usepackage[labelfont=bf,textfont=it,font=small]{caption}

\usepackage{algorithm}
\usepackage{algorithmicx}
\usepackage[noend]{algpseudocode}
\algrenewcommand\algorithmiccomment[1]{\hfill{\color{gray}$\triangleright$~#1}}

\usepackage{dsfont}
\newcommand{\ones}{\mathds{1}}

\usepackage{xcolor}
\definecolor{c0}{HTML}{641a80}
\definecolor{c1}{HTML}{b73779}

\usepackage{tikz}
\usetikzlibrary {graphs} \usepackage{comment}

\usepackage{hyperref}
\usepackage[capitalize,nameinlink,noabbrev]{cleveref}
\crefformat{equation}{#2(#1)#3}
\crefmultiformat{equation}{(#2#1#3)}%
{ and~#2(#1)#3}{, #2(#1)#3}{ and~#2(#1)#3}

\crefname{problem}{Problem}{Problems}
\crefname{fact}{Fact}{Facts}

\usepackage{hyperref}
\hypersetup{
    colorlinks,
    linkcolor=c0,
    citecolor=c1,
    urlcolor=c0
}
\renewcommand{\vec}{\mathbf}

\newcommand{\EE}{\operatorname{\mathbb{E}}}
\newcommand{\PP}{\operatorname{\mathbb{P}}}
\newcommand{\VV}{\operatorname{\mathbb{V}}}
\newcommand{\R}{\mathbb{R}}

\newcommand{\T}{\mathsf{T}}
\newcommand{\F}{\mathsf{F}}
\newcommand{\comp}{\mathsf{c}}

\newcommand{\diag}{\operatorname{diag}}

\renewcommand{\d}{\mathrm{d}}
\newcommand{\tr}{\operatorname{tr}}

\newcommand{\circdiv}{\ooalign{\hss$\circ$\hss\cr$\div$}}

\usepackage[backend=biber,backref,style=alphabetic,maxalphanames=4,maxcitenames=99,maxbibnames=99,giveninits=true,url=false,doi=false,isbn=false,eprint=false,date=year]{biblatex}
\DeclareSourcemap{
  \maps[datatype=bibtex, overwrite]{
    \map{
      \step[fieldset=editor, null]
    }
  }
}
\DefineBibliographyStrings{english}{%
  backrefpage = {cited on page},%
  backrefpages = {cited on pages}%
}

\allowdisplaybreaks

\title{Fixed-sparsity matrix approximation from \\matrix-vector products}
\author{
    Noah Amsel \\ New York University\\ \texttt{noah.amsel@nyu.edu}
	\and 
	Tyler Chen \\ New York University\\ \texttt{tyler.chen@nyu.edu}
	\and 
    Feyza Duman Keles \\ New York University\\ \texttt{fd2135@nyu.edu}
	\and 
    Diana Halikias \\ Cornell University\\ \texttt{dh736@cornell.edu}
	\and 
	Cameron Musco\\ UMass Amherst\\ \texttt{cmusco@cs.umass.edu}
	\and
	Christopher Musco\\ New York University\\ \texttt{cmusco@nyu.edu}
}

\usepackage{fancyhdr}
\date{}

\begin{document}

\maketitle

\begin{abstract}
    We study the problem of approximating a matrix $\vec{A}$ with a matrix that has a fixed sparsity pattern (e.g., diagonal, banded, etc.), when $\vec{A}$ is accessed only by matrix-vector products. 
    We describe a simple randomized algorithm that returns an approximation with the given sparsity pattern with Frobenius-norm error at most $(1+\varepsilon)$ times the best possible error.
    When each row of the desired sparsity pattern has at most $s$ nonzero entries, this algorithm requires  $O(s/\varepsilon)$ non-adaptive matrix-vector products with $\vec{A}$.
    We also prove a matching lower-bound, showing that, for any sparsity pattern with $\Theta(s)$ nonzeros per row and column, any algorithm achieving $(1+\epsilon)$ approximation requires $\Omega(s/\varepsilon)$ matrix-vector products in the worst case.
    We thus resolve the matrix-vector product query complexity of the problem up to constant factors, even for the well-studied case of diagonal approximation, for which no previous lower bounds were known.
\end{abstract}

\section{Introduction}

Learning about a matrix $\vec{A}$ from matrix-vector product (matvec) queries\footnote{In some settings, it may also be possible to do matrix-transpose-vector queries $\vec{y}\mapsto \vec{A}^\T\vec{y}$. It is an open question to understand when matrix-transpose-vector queries may be beneficial \cite{boulle_halikias_otto_townsend_24}.}  
$\vec{x}_1, \ldots, \vec{x}_m \mapsto \vec{A}\vec{x}_1, \ldots, \vec{A}\vec{x}_m$
is a widespread task in numerical linear algebra, theoretical computer science, and machine learning. Algorithms that only access $\vec A$ through matvec queries (also known as \emph{matrix-free} algorithms) are useful when $\vec{A}$ is not known explicitly, but admits efficient matvecs. Common settings include when $\vec A$ is a solution operator for a linear partial differential equation~\cite{karniadakis_kevrekidis_lu_perdikaris_wang_yang_21, schafer_owhadi_21, boulle_earls_townsend_22}, 
a function of another matrix that can be applied with Krylov subspace methods \cite{gallopoulos_saad_92,higham_08}, the Hessian of a neural network that can be applied via back-propagation \cite{pearlmutter_94}, or a data matrix in compressed sensing applications \cite{kalender_kalender_11}.
Matrix-free algorithms are often the methods of choice due to practical considerations such as memory usage and data movement. If the queries are chosen \emph{non-adaptively} --- i.e, the $i$-th query vector $\vec{x}_i$ does not depend on the results $\vec{Ax}_1,\ldots,\vec{Ax}_{i-1}$ from previous queries ---  queries can be  parallelized, potentially leading to substantial runtime improvements \cite{murray_demmel_mahoney_erichson_melnichenko_malik_grigori_luszczek_derezinski_lopes_liang_luo_dongarra_23}. Non-adaptive queries can also be computed in a single pass over $\vec A$, and thus, non-adaptive matvec query algorithms are an important tool in streaming and distributed computation \cite{clarkson_woodruff_09,martinsson_tropp_20}.

In many cases, the cost of matrix-free algorithms is dominated by the cost of the matvec queries.
As such, a key goal is to understand the minimum number of queries required to solve a given problem, also known as the query complexity. Any algorithm automatically yields an upper bound on the query complexity, whereas it can be more challenging to prove  lower bounds.
Problems for which the matvec query complexity have been extensively studied include 
low-rank approximation \cite{halko_martinsson_tropp_11,simchowitz_elalaoui_recht_18,tropp_webber_23,bakshi_clarkson_woodruff_22,bakshi_narayanan_23, halikias_townsend_23}, 
spectrum approximation \cite{skilling_89,weisse_wellein_alvermann_fehske_06,swartworth_woodruff_23},
trace and diagonal estimation \cite{girard_87,hutchinson_89,bekas_kokiopoulou_saad_07,meyer_musco_musco_woodruff_21}, the approximation of linear system solutions and the action of matrix functions \cite{gallopoulos_saad_92,greenbaum_97,braverman_hazan_simchowitz_woodworth_20}, and matrix property testing \cite{sun_woodruff_yang_zhang_21,needell_swartworth_woodruff_22}.

An important class of problems considers recovering $\vec A$ from matvec queries when $\vec A$ is  structured, or relatedly, finding the nearest approximation to $\vec A$ within a structured class of matrices. Example classes that have been studied extensively in this setting include low-rank matrices \cite{halko_martinsson_tropp_11, tropp_webber_23}, hierarchical low-rank matrices \cite{lin_lu_ying_11, martinsson_16, levitt_martinsson_22,levitt_martinsson_22a, schafer_owhadi_21, halikias_townsend_23}, diagonal matrices \cite{bekas_kokiopoulou_saad_07,tang_saad_11,baston_nakatsukasa_22,dharangutte_musco_23}, sparse matrices \cite{curtis_powell_reid_74,coleman_more_83,coleman_cai_86, wimalajeewa_eldar_varshney_13, dasarathy_shah_bhaskar_nowak_15}, and beyond \cite{waters_sankaranarayanan_baraniuk_11, schafer_katzfuss_owhadi_21}.

\subsection{Fixed-sparsity matrix approximation}

In this work, we focus on the task of approximating $\vec{A}$ with a matrix of a \emph{specified} sparsity pattern, with error competitive with the best approximation of the given sparsity pattern. 
This is a natural task; indeed, there are many existing linear algebra algorithms for matrices of a given sparsity pattern (e.g. diagonal, tridiagonal, banded, block diagonal, etc.), so it is a common goal to obtain an approximation compatible with such algorithms. 
As we discuss in \cref{sec:past}, several important special cases including diagonal approximation and exact recovery of matrices with known-sparsity have been studied extensively in prior work.

Formally, using ``$\,\circ\,$'' to indicate the Hadamard (entrywise) product and $\| \cdot \|_\F$ to denote the Frobenius norm, we consider the following problem: 
\begin{problem}[Best approximation by a matrix of fixed sparsity]\label{prob:recovery_full}
    Given a matrix $\vec{A} \in\mathbb{R}^{n\times d}$ and a binary matrix $\vec{S} \in \{0,1\}^{n\times d}$, find a matrix $\widetilde{\vec{A}}$ so that $\widetilde{\vec{A}} = \vec{S}\circ\widetilde{\vec{A}}$ and
    \[
    \| \vec{A} - \widetilde{\vec{A}} \|_\F
    \leq (1+\varepsilon) \, \| \vec{A} - \vec{S}\circ \vec{A} \|_\F.
    \]
\end{problem}
Observe that $\vec{S}\circ \vec{A}$ is the matrix of sparsity $\vec{S}$ nearest to $\vec{A}$ in the Frobenius norm:
\[
\vec{S}\circ\vec{A}
= \operatornamewithlimits{argmin}_{\vec{X} = \vec{S}\circ\vec{X}} \| \vec{A} - \vec{X} \|_\F.
\]
Hence, \cref{prob:recovery_full} is asking for a \emph{near-optimal} approximation to $\vec{A}$ of the given sparsity $\vec{S}$.\footnote{This is reminiscent of the low-rank approximation problem, in which we aim to find a rank-$k$ approximation to $\vec{A}$ competitive with the best rank-$k$ approximation \cite{halko_martinsson_tropp_11,tropp_webber_23}.}
In particular, if $\vec A$ already has sparsity pattern $\vec{S}$, then $\vec{A} = \vec{S}\circ \vec{A}$, so solving \cref{prob:recovery_full} will recover $\vec{A}$ exactly. 

In addressing \cref{prob:recovery_full}, it will also be beneficial to consider the closely related problem of recovering the ``sparse-part'' of a matrix:
\begin{problem}[Best approximation to on-sparsity-pattern entries]\label{prob:recovery_sparse}
Given a matrix $\vec{A} \in\mathbb{R}^{n\times d}$ and a binary matrix $\vec{S} \in \{0,1\}^{n\times d}$, find a matrix $\widetilde{\vec{A}}$ so that $\widetilde{\vec{A}} = \vec{S}\circ\widetilde{\vec{A}}$ and
    \[
    \| \vec{S}\circ \vec{A} - \widetilde{\vec{A}} \|_\F^2
    \leq \varepsilon \,\| \vec{A} - \vec{S}\circ \vec{A} \|_\F^2.
    \]
\end{problem}
Note the presence of the squared norms in \cref{prob:recovery_sparse}. 
Since $\widetilde{\vec{A}}$ has the same sparsity as $\vec{S}$; i.e. $\widetilde{\vec{A}} = \vec{S}\circ\widetilde{\vec{A}}$,
\begin{equation}
\label{eqn:AAtilde-decomp}
\| \vec{A} - \widetilde{\vec{A}} \|_\F^2
= \| \vec{A} - \vec{S}\circ \vec{A} \|_\F^2 + \|\vec{S}\circ\vec{A} - \widetilde{\vec{A}}\|_\F^2.
\end{equation}
Thus, using the fact that $\sqrt{1+2\varepsilon} < 1+\varepsilon$ for all $\varepsilon>0$, a solution to \cref{prob:recovery_sparse} with accuracy ${2\varepsilon}$ immediately yields a solution to \cref{prob:recovery_full} with accuracy $\varepsilon$. 
Conversely, if $\varepsilon\in(0,1)$ so that $\sqrt{1+3\varepsilon} \geq 1+\varepsilon$, then a solution to \cref{prob:recovery_full} with accuracy $\varepsilon$ yields a solution to \cref{prob:recovery_sparse} with accuracy ${3\varepsilon}$.
In this sense, the problems are equivalent.

\subsection{Our Contributions and Roadmap}

Our first contribution is to analyze a simple algorithm (\cref{alg:main}) that solves \cref{prob:recovery_full,prob:recovery_sparse}. When the sparsity pattern $\vec S$ has at most $s$ non-zero entries per row, this algorithm uses $m = O(s/\varepsilon)$ non-adaptive matrix-vector product queries.
Specifically, the algorithm computes $\vec{Z} = \vec{A}\vec{G}$, where $\vec{G}$ is a $d\times m$ matrix with independent standard normal entries, and then outputs the matrix 
\begin{equation}\label{eqn:algmain}
    \widetilde{\vec{A}} = \operatornamewithlimits{argmin}_{\vec{X} = \vec{S}\circ\vec{X}} \| \vec{Z} - \vec{X} \vec{G}\|_\F.
\end{equation}
We make no claims about the novelty of the algorithm, which follows immediately from ideas in compressed sensing.
In~\cref{sec:alg}, using standard tools from random matrix theory and high dimensional probability, we provide an analysis of \cref{alg:main} and prove the following:
\begin{restatable}{theorem}{ubmain}
\label{thm:ub_main}
Consider any $\vec{A} \in\mathbb{R}^{n\times d}$ and any $\vec{S}\in\{0,1\}^{n\times d}$ with at most $s$ nonzero entries per row. 
Then, for any $m\geq s+2$, using $m$ randomized matrix-vector queries, \cref{alg:main} returns a matrix $\widetilde{\vec{A}}$, equal to $\vec{S}\circ\vec{A}$ in expectation, satisfying
\[
\EE\Bigl[\|\vec{S}\circ\vec{A} - \widetilde{\vec{A}}\|_\F^2 \Bigr]
\leq {\frac{s}{m-s-1}} \| \vec{A} - \vec{S}\circ \vec{A} \|_\F^2.
\]
The above inequality is equality if each row of $\vec{S}$ has exactly $s$ non-zero entries.\end{restatable}

Owing to \cref{eqn:AAtilde-decomp} and Jensen's inequality, \cref{thm:ub_main} also gives an expectation bound for \cref{prob:recovery_full}.
Setting $m = O(s/\varepsilon)$ implies that \cref{prob:recovery_full,prob:recovery_sparse} are solved in expectation. 
Using Markov's inequality, we derive a probability bound:
\begin{restatable}{corollary}{ubmainprob}
\label{thm:ub_main_prob}
In the setting of \cref{thm:ub_main}, for any $\varepsilon>0$ and $\delta \in (0,1)$, if $m\geq s+2$ then
\[
m\geq s \left(\frac{1}{2\delta \varepsilon} +1 \right) + 1
\quad\Longrightarrow\quad
\PP\Bigl[\|\vec{A} - \widetilde{\vec{A}}\|_\F \geq (1+\varepsilon) \, \| \vec{A} - \vec{S}\circ \vec{A} \|_\F \Bigr] 
\leq \delta.
\]
\end{restatable}
Hence, using $O(s/\varepsilon)$ \cref{alg:main} solves \cref{prob:recovery_full} except with small constant probability, say $\le 1/100$.

We proceed, in \cref{sec:lower_bounds}, to study lower bounds for the matvec query complexity of \cref{prob:recovery_full,prob:recovery_sparse}.
We show that up to constant factors, the upper bound in \cref{thm:ub_main_prob} is optimal, even for the stronger class of adaptive matvec query algorithms:
\begin{restatable}{theorem}{lbmain}
\label{thm:lb_main}
Fix $\gamma\in(0,1)$.
Then there exist constants $c,C>0$ (depending only on $\gamma$) such that the following holds:

For any $\varepsilon\in(0,c)$ and integer $s\geq 1$, there is a distribution on (symmetric) matrices $\vec{A} \in \R^{d \times d}$ such that, for any sparsity pattern $\vec{S}$ whose rows and columns each have between $\gamma s$ and $s$ nonzero entries, and for any (possibly randomized) algorithm that uses $m < Cs/\varepsilon$ (possibly adaptive) matrix-vector queries to $\vec{A}$ to output $\widetilde{\vec{A}}$ with $\widetilde{\vec{A}} \circ \vec S = \widetilde{\vec{A}}$,
\[
    \PP\Bigl[ \| \vec{A} - \widetilde{\vec{A}} \|_\F
    \leq (1+\varepsilon) \,\| \vec{A} - \vec{S}\circ \vec{A} \|_\F \Bigr] \leq \frac{1}{25}.
    \]
\end{restatable} 
Thus, $\Omega(s/\varepsilon)$ queries are required to solve \cref{prob:recovery_full} with any reasonable probability.
Note that since solving \cref{prob:recovery_full} gives a solution to \cref{prob:recovery_sparse}, the lower bound \cref{thm:lb_main} also implies that $\Omega(s/\varepsilon)$ queries are required to solve \cref{prob:recovery_sparse}. 
Our approach uses an invariance property of Wishart matrices after a sequence of adaptive queries \cite{braverman_hazan_simchowitz_woodworth_20}.
Note that since our hard instance is symmetric, \cref{thm:lb_main} also holds against algorithms which are allowed to use matvec queries with $\vec{A}^\T$.

Our lower bound applies even to the well-studied case of best approximation by a diagonal matrix, and more broadly, to approximation by a banded matrix. 
To the best of our knowledge, our lower bounds are the first (adaptive or non-adaptive) for \cref{prob:recovery_full,prob:recovery_sparse} even for these special cases.

In \cref{sec:coloring}, we compare our algorithm to widely-used coloring-based methods for fixed-sparsity pattern matrix approximation \cite[etc.]{curtis_powell_reid_74,coleman_more_83,coleman_cai_86}.
We show that there are situations where coloring methods perform worse than the algorithm described in this paper by a quadratic factor or more. Specifically, even for matrices with $\leq s$ non-zeros per row \emph{and} column, they can require $O(s^2)$ instead of $O(s)$ matvec queries.   
We also discuss a setting in which coloring methods can outperform the algorithm from this paper.

Finally, in \cref{sec:numerical} we present several numerical experiments on test problems to illustrate the sharpness of our upper bound, and in \cref{sec:outlook} we discuss the potential for future work, including the potential for algorithms which combine the algorithm in this paper with coloring-algorithms to obtain more robustness to noise.

\subsection{Past work}\label{sec:past}

To the best of our knowledge, \cref{prob:recovery_full,prob:recovery_sparse} have not been previously stated explicitly in the given generality. However, there is a range of past work that studies special cases of these problems. 
We categorize these works into the zero error case, where $\vec{A} = \vec{S}\circ\vec{A}$, and the nonzero error case, where $\vec{A}\neq \vec{S}\circ\vec{A}$.
In addition, we discuss how \cref{prob:recovery_full,prob:recovery_sparse} differ from the standard sparse-recovery problem in compressed sensing. 
\subsubsection{Zero error}

Some of the earliest work relating to \cref{prob:recovery_full,prob:recovery_sparse} seeks to recover Jacobian and Hessian matrices with a known sparsity pattern from matvecs \cite{curtis_powell_reid_74,coleman_more_83,coleman_cai_86}.
These methods make use of the fact that a \emph{graph coloring} of a particular graph, induced by the sparsity pattern of $\vec{A}$, can be used to obtain a set of query vectors which are sufficient to exactly recover $\vec{A}$. In many (but not all) cases, exact recovery is possible with $s$ queries.
Coloring methods have also influenced many algorithms for the nonzero error setting.
We compare our results to such coloring-based methods in \cref{sec:coloring}, arguing that \cref{alg:main} always performs better in the zero-error setting, since it always requires just $s$ queries.

More generally, \cite{halikias_townsend_23} studies the matvec query complexity of exact recovery of a wide range of linearly parameterized matrix families, proving matching upper and lower-bounds on the number of queries required. 
In particular, it is shown that recovering a diagonal matrix requires  one query, recovering a block-diagonal matrix with $s\times s$ blocks requires  $s$ queries, and recovering a tridiagonal matrix requires  $3$ queries. Recovering a general $n \times d$ matrix is shown to require $d$ queries.
With probability one, \cref{alg:main} matches these lower bounds (see \cref{thm:exact_sparse}). 

\subsubsection{Nonzero error}

The most theoretically well-studied instance of \cref{prob:recovery_sparse} is arguably the case $\vec{S} = \vec{I}$; i.e. the task of approximating the diagonal of a matrix.
For this task, it is common to use Hutchinson's diagonal estimator, defined as
\begin{equation}\label{eqn:hutch_diag}
\vec{d}_m = \Big[ \sum_{j=1}^{m} \vec{r}_j \circ (\vec{A}\vec{r}_j) \Big] \circdiv 
\Big[ \sum_{j=1}^{m} \vec{r}_j \circ \vec{r}_j \Big]
,
\end{equation}
where ``$\circdiv$'' indicates entrywise division and the entries of the vectors $\vec{r}_j$ are all independent random variables with mean zero and variance one.
A number of analyses of this estimator have been given for various distributions \cite{bekas_kokiopoulou_saad_07,tang_saad_11,baston_nakatsukasa_22,hallman_ipsen_saibaba_23,dharangutte_musco_23}.
In particular, \cite{baston_nakatsukasa_22,dharangutte_musco_23} give error bounds for \cref{prob:recovery_sparse}, showing it suffices to set $m = O(1/\varepsilon)$, matching \cref{thm:ub_main} up to constant factors. In fact, when $\vec S = \vec I$ our \cref{alg:main}  is equivalent to \cref{eqn:hutch_diag} if the query vectors $\vec{r}_i$ are Gaussian.
We detail this connection in \cref{rem:diag_equiv}.

Past work has also studied~\cref{prob:recovery_full} with the goal of approximating a potentially non-sparse matrix by a sparse matrix.
For instance, there is a long line of work on approximating matrix functions $\vec{A} = f(\vec{H})$ from matvecs.
If $\vec{H}$ is sparse, then the entries of $\vec{A} = f(\vec{H})$ decay exponentially away from the nonzero entries of $\vec{H}$ under mild assumptions on $f(x)$~\cite{demko_moss_smith_84,benzi_razouk_08,benzi_boito_razouk_13}. 
As such, it is reasonable to approximate $\vec{A}$ with a sparse matrix of a sparsity similar to $\vec{H}$.
This observation has been used in matrix approximation algorithms~\cite{tang_saad_11,stathopoulos_laeuchli_orginos_13,frommer_schimmel_schweitzer_21,park_nakatsukasa_23}.
Broadly speaking, these algorithms aim to combine the coloring methods described above with the estimator $\vec{d}_m$ described in \cref{eqn:hutch_diag}.
For banded matrices, one can use the existing analyses of $\vec{d}_m$ to analyze the performance of these methods, showing that they solve \cref{prob:recovery_sparse} to accuracy ${\varepsilon}$ using $O(s/\varepsilon)$ matrix-vector queries.
We include a note on this in \cref{sec:banded}, as we were unable to find such an analysis in the literature.
\Cref{thm:ub_main} shows that \Cref{alg:main} matches this bound for arbitrary sparsity patterns.

More recently, motivated by the field of partial differential equation (PDE) learning, there has been widespread interest in learning the solution operators of PDEs from input-output data of forcing terms and solutions, analogous to matrix-vector products~\cite{schafer_katzfuss_owhadi_21,schafer_owhadi_21,boulle_halikias_townsend_23,karniadakis_kevrekidis_lu_perdikaris_wang_yang_21}. 
The method in \cite{schafer_owhadi_21} obtains a fixed-sparsity approximation to the sparse Cholesky factorization of the solution operator by coloring, which is provably accurate for certain problems.
This makes use of the fact that in certain settings a fixed-sparsity Cholesky factorization is accurate and can be efficiently computed   \cite{schafer_katzfuss_owhadi_21}.
This is broadly related to the (factorized) sparse approximate inverse problem for obtaining preconditioners \cite{benzi_tuma_99}. The method in~\cite{boulle_halikias_townsend_23} also derives a continuous analogue of a generalized coloring algorithm for targeting low-rank subblocks of hierarchical matrices~\cite{levitt_martinsson_22a}, and then recovering these subblocks using the randomized SVD~\cite{halko_martinsson_tropp_11}. The final step of this algorithm reduces to the recovery of a block diagonal matrix.

\subsubsection{The sparse recovery problem in compressed-sensing}
\label{sec:sparse-recovery}

It is important to contrast the aims and methods of this paper with the rich literature on compressed sensing and  sparse recovery \cite[etc.]{eldar_kutyniok_12,foucart_13}.

Given access to a length $d$ vector $\vec{a}$ through linear measurements of $\vec{a}$: $\vec{M}\mapsto \vec{M}\vec{a}$, the goal of the $\ell_2/\ell_2$ sparse recovery problem is to obtain an $s$-sparse vector $\widetilde{\vec{a}}$ for which 
\[
\|\vec{a} - \widetilde{\vec{a}} \|_2 \leq (1+\varepsilon) \min_{\vec{a}'\text{ $s$-sparse}} \| \vec{a} - \vec{a}'\|_2.
\]
Critically,  the support of $\widetilde{\vec{a}}$ is not known ahead of time.
In fact, if the support were specified, this would be a trivial problem; simply take $\vec{M}$ to have $s$ rows, each a standard basis vector corresponding to an entry  of the support.

A number of past works \cite[etc.]{waters_sankaranarayanan_baraniuk_11,wimalajeewa_eldar_varshney_13,dasarathy_shah_bhaskar_nowak_15} have also studied a matrix version of this problem in which one aims to obtain an $sd$-sparse matrix $\widetilde{\vec{A}}$ for which 
\[
\|\vec{A} - \widetilde{\vec{A}} \|_\F \leq (1+\varepsilon) \min_{\vec{X}\text{ $sd$-sparse}} \| \vec{A} - \vec{X}\|_\F,
\]
using only bi-linear measurements of $\vec{A}$: $(\vec{U}, \vec{V}) \mapsto \vec{U}\vec{A}\vec{V}^\T$. 
Note that the matrix problem is actually equivalent to a restricted version of the vector problem. 
Indeed, if $\vec{a} = \operatorname{vec}(\vec{A})$ and $\vec{M} = \vec{U}\otimes \vec{V}$, then $\vec{M}\vec{a} = \operatorname{vec}(\vec{U}\vec{A}\vec{V}^\T)$. 
Here $\operatorname{vec}(\cdot)$ forms a vector by stacking the columns of a matrix on top of one another and $\otimes$ is the Kronecker product.
In the zero error setting it has been shown  that if $\vec{A}$ is an $n \times n$ matrix with sufficiently  distributed sparsity, one can use a convex program to stably recover $\vec A$ using $O (\sqrt{ \operatorname{nnz}(\vec A) \cdot \log n})$ queries on each side~\cite{dasarathy_shah_bhaskar_nowak_15}.

Even for the general vector recovery problem, such algorithms necessarily have worse dependencies on $\varepsilon$ and $d$ than the bounds we prove for the algorithm described in the next section.
In particular, algorithms solving the $\ell_2/\ell_2$  sparse recovery problem  necessarily require $\Omega(s / \varepsilon^2 + s\log(ds) / \varepsilon)$ linear measurements \cite{price_woodruff_11}.
In fact, even if the problem is relaxed, so that the output vector $\widetilde{\vec{a}}$ is allowed to be non-sparse, $\Omega(s \log(d/s) / \varepsilon )$ queries are required \cite{price_woodruff_11}.
In contrast, our upper-bound \cref{thm:ub_main}/\cref{thm:ub_main_prob} have no dependence on the dimension $d$.

\subsection{Notation}

For a set $S$, $S^\comp$ indicates the complement (determined from context).
For $d \geq1$, we define $[d] = \{1,2, \ldots, d\}$.
For $R\subset [n]$ and $C\subset[d]$, $[\vec{X}]_{R,C}$ indicates the $|R|\times |C|$ submatrix of a $n\times d$ matrix $\vec{X}$ corresponding to the rows in $R$ and columns in $C$. If $R$ or $C$ contain only one element, we will simply write this element. Likewise, when $R = [n]$ or $C = [d]$, we will use a colon; e.g. $[\vec{X}]_{1,:}$ is the first column of $\vec{X}$.

We denote the Frobenius norm of a matrix $\vec{X}$ by $\| \vec{X} \|_\F$, the transpose by $\vec{X}^\T$, and the pseudo-inverse by $\vec{X}^\dagger$.
We use ``\,$\circ$\,'' to denote the Hadamard (entrywise) product. 
Specifically, for matrices $\vec{X}$ and $\vec{Y}$, $\vec{X}\circ\vec{Y}$ is the matrix defined by $[\vec{X}\circ\vec{Y} ]_{i,j} = [\vec{X}]_{i,j} [\vec{Y}]_{i,j}$.
We use $\vec{0}$ and $\vec{1}$ to denote matrices of all zeros or ones, with size determined from context. 

Throughout $\PP$ will be used to indicate probabilities, $\EE$ the expectation of a random variable, and $\VV$ the variance.
We denote by $\mathcal{N}(\mu,\sigma^2)$ the Gaussian distribution with mean $\mu$ and variance $\sigma^2$. We use $\operatorname{Gaussian}(n,d)$ (or $\operatorname{Gaussian}(d)$ if $n=d$) to denote the distribution on $n\times d$ matrices, where each entry of the matrix is independent and identically distributed (iid) with distribution $\mathcal{N}(0,1)$.

\section{An algorithm and upper bound}\label{sec:alg}

We begin by writing down an explicit algorithm (\cref{alg:main}) for solving \cref{prob:recovery_full,prob:recovery_sparse}
This algorithm proceeds row-by-row, taking a advantage of the fact that different rows of the solution to \cref{eqn:algmain} do not depend on one another (except through the common use of $\vec{Z} = \vec{A}\vec{G}$).
For each row, we can solve for the entries of $\widetilde{\vec{A}}$ via an appropriate least squares problem.
\begin{algorithm}[ht]
\caption{Fixed-sparse-matrix recovery}\label{alg:main}
\fontsize{10}{14}\selectfont
\begin{algorithmic}[1]
\Procedure{fixed-sparse-matrix-recovery}{$\vec{A},\vec{S},m$}
\State Form $\vec{G} \sim \operatorname{Gaussian}(d,m)$ \Comment{$d\times m$ iid Gaussian matrix}
\State Compute $\vec{Z} = \vec{A}\vec{G}$ \Comment{$m$ non-adaptive matvec queries}
\For{$i=1,2,\ldots, n$}
\State Let $S_i = \{ j : [\vec{S}]_{i,j} = 1 \}$  \Comment{nonzero entries of $i$th row of $\vec{S}$}
\State Let $\vec{z}_i^\T = [\vec{Z}]_{i,:}$  \Comment{$i$-th row of $\vec{Z}$}
\State Let $\vec{G}_i^\T = [\vec{G}]_{S_i,:}$ \Comment{submatrix formed by taking the rows from $S_i$}
\State Compute $\widetilde{\vec{a}}_i = \vec{G}_i^\dagger \vec{z}_i$
\Comment{solve a $m\times|S_i|$ least squares problem}
\State Set $[\widetilde{\vec{A}}]_{i,S_i} = \widetilde{\vec{a}}_i^\T$ and $[\widetilde{\vec{A}}]_{i,S_i^\comp} = \vec{0}^\T$
\Comment{construct $i$th row of $\widetilde{\vec{A}}$}
\EndFor
\State \Return $\widetilde{\vec{A}}$
\EndProcedure
\end{algorithmic}
\end{algorithm}

Note that in the case that $\vec{A}$ has nonzeros only in positions where $\vec{S}$ is nonzero, \cref{alg:main} will exactly recover $\vec{A}$ as long as the least squares problem for each row is fully determined. 

\begin{proposition}\label{thm:exact_sparse}
    If $\vec{S}\circ \vec{A} = \vec{A}$ and $m\geq s$, then \cref{alg:main} returns a matrix $\widetilde{\vec{A}} = \vec{A}$ with probability one.
\end{proposition}

\begin{proof}
    Consider the $i$-th row, and let $\vec x_i \in \R^{|S_i|}$ be the set of non-zero entries in that row. The corresponding $\vec{G}_i \in \mathbb{R}^{m\times |S|}$ is full rank $|S_i|$ with probability one.
    Observe that $\vec{z}_i = \vec{G}_i \vec{x}_i$.
    Thus, since $\vec{G}_i$ is full-rank, $\widetilde{\vec{a}}_i = \vec{G}_i^\dagger \vec z = \vec{G}_i^\dagger \vec{G}_i \vec x_i =\vec x_i$. 
    Thus, we recover the row exactly.
    By a union bound, with probability one, this simultaneously happens for all the rows.
\end{proof}

Our main focus will be the case where $\vec{A}$ may have nonzeros off of the specified sparsity pattern.
We first recall a standard result from high dimensional probability.
\begin{proposition}\label{thm:gaussian_expectations}
    Let $\vec{G} \sim \operatorname{Gaussian}(p,q)$. 
    Then, for compatible matrices $\vec{X}$ and $\vec{Y}$, 
    \begin{equation}\label{eqn:gaussian_frob_prod}
        \EE\bigl[ \| \vec{X} \vec{G} \vec{Y} \|_\F^2 \bigr]
        = \| \vec{X} \|_\F^2 \|\vec{Y}\|_\F^2.        
    \end{equation}
    Moreover, if $p-q \geq 2$, then 
    \begin{equation}\label{eqn:gaussian_inv_frob}
    \EE\bigl[ \| \vec{G}^\dagger \|_\F^2 \bigr]
    = \frac{q}{p-q-1}.
    \end{equation}
\end{proposition}

\begin{proof}
The expression \cref{eqn:gaussian_frob_prod} is an elementary calculation; see for instance \cite[Proposition A.1]{halko_martinsson_tropp_11}.
The expression \cref{eqn:gaussian_inv_frob} follows from the fact that $\vec{G}^\T\vec{G}$ is invertible with probability one, and $(\vec{G}^\T \vec{G})^{-1}$ has a inverse Wishart distribution, which, for $p-q\geq 2$ has mean $\vec{I} / (p-q-1)$ \cite[\S3.2 (12)]{muirhead_82}. 
Since $\|\vec{G}^\dagger\|_\F^2 = \tr((\vec{G}^\dagger)^\T (\vec{G}^\dagger)) = \tr((\vec{G}^\dagger)^\T (\vec{G}^\dagger))$, the result follows from the linearity of the expectation; see for instance \cite[Proposition A.5]{halko_martinsson_tropp_11}.
\end{proof}

Using \cref{thm:gaussian_expectations}, we establish the following theorem.
\ubmain*
We also obtain a probability bound:
\ubmainprob* 

\begin{proof}[Proof of \cref{thm:ub_main}]
The algorithm processes $\vec{Z} = \vec{A}\vec{G} \in \mathbb{R}^{n\times m}$ sequentially to approximate the rows of $\vec{A}$. 
Fix $i$ and let $S_i$ be the indices of the nonzero entries of $[\vec{S}]_{i,:}$ and  $\vec{z}_i^\T = [\vec{Z}]_{i,:}$ be the $i$-th row of $\vec{Z}$ and $S_i^\comp = [d]\setminus S_i$.
Let $\vec{G}_i^\T = [\vec{G}]_{S_i,:}$ and $\widehat{\vec{G}}_i^\T = [\vec{G}]_{S_i^\comp,:}$ be submatrices of $\vec{G}$ formed by taking the rows of $\vec{G}$ in $S_i$ and $S_i^\comp$ respectively.
Define $\vec{x}_i^\T  = [\vec{A}]_{i,S_i}$ and $\vec{y}_i^\T = [\vec{A}]_{i,S_i^\comp}$ and observe that 
\[
\vec{z}_i^\T := [\vec{A}\vec{G}]_{i,:}
=
\vec{x}_i^\T \vec{G}_i^\T + \vec{y}_i^\T \widehat{\vec{G}}_i^\T.
\]

To enforce the sparsity pattern, \cref{alg:main} tries to recover $\vec{x}_i\in\mathbb{R}^s$ from $\vec{z}_i\in\mathbb{R}^{m}$ by solving the least squares problem:
\[
\widetilde{\vec{a}}_i
:= \vec{G}_i^\dagger \vec{z}_i
= \vec{G}_i^\dagger (\vec{G}_i \vec{x}_i + \widehat{\vec{G}}_i \vec{y}_i)
= \vec{x}_i + \vec{G}_i^\dagger \widehat{\vec{G}}_i \vec{y}_i.
\]
Here we have used that $\vec{G}_i$ is full-rank with probability one.

Since $\vec{G}_i$ and $\widehat{\vec{G}}_i$ are independent, clearly $\EE[\widetilde{\vec{a}}_i] = \vec{x}_i$ as $\EE[\widehat{\vec{G}}_i] = \vec{0}$. Thus, \cref{alg:main} outputs an unbiased estimator for $\vec{S}\circ\vec{A}$.

As long as $m\geq |S_i|+2$, it follows from standard results in random matrix theory that
\begin{align*}
\EE\bigl[ \| \vec{x}_i - \widetilde{\vec{a}}_i \|_2^2 \bigr]
&= \EE\bigl[ \| \vec{G}_i^\dagger \widehat{\vec{G}}_i \vec{y}_i\|_2^2 \bigr]
\\&= \EE\bigl[\EE\bigl[ \| \vec{G}_i^\dagger \widehat{\vec{G}}_i \vec{y}_i\|_2^2 \, \big|\, \vec{G}_i \bigr]\bigr] 
\\&= \EE\bigl[\| \vec{G}_i^\dagger\|_\F^2 \cdot  \| \vec{y}_i\|_2^2\bigr] \tag*{\cref{eqn:gaussian_frob_prod} in \cref{thm:gaussian_expectations}}
\\&= \frac{|S_i|}{m-|S_i|-1} \cdot \| \vec{y}_i \|_2^2 \tag*{\cref{eqn:gaussian_inv_frob} in \cref{thm:gaussian_expectations}}
\\&\leq \frac{s}{m-s-1} \cdot \| \vec{y}_i \|_2^2,
\end{align*}
where we have used that $|S_i|\leq s$ in the final line (and hence we have equality if $|S_i| = s$). 

Let $\widetilde{\vec{A}}$ be the output of \cref{alg:main}.
Then, by the linearity of expectation,
\begin{align*}
\EE\Bigl[\|\vec{S}\circ\vec{A} - \widetilde{\vec{A}}\|_\F^2\Bigr]
&= \sum_{i=1}^{n} \EE\Bigl[\| \vec{x}_i - \widetilde{\vec{a}}_i \|_2^2\Bigr] 
\\&\leq \frac{s}{m-s-1} \sum_{i=1}^{n} \| \vec{y}_i \|_2^2 
\\&= \frac{s}{m-s-1} \| \vec{A} - \vec{S}\circ \vec{A} \|_\F^2.
\end{align*}
Observe that we have equality if $|S_i|=s$ for each row $i\in[n]$.
\end{proof}

\begin{proof}[Proof of \cref{thm:ub_main_prob}]
Applying Markov's inequality to \cref{thm:ub_main}, we find
\begin{equation}\label{eqn:ub_main:SAA}
    \PP\bigl[\|\vec{S}\circ\vec{A} - \widetilde{\vec{A}}\|_\F^2 \geq \alpha  \bigr]
\leq \frac{s}{m-s-1} \frac{\| \vec{A} - \vec{S}\circ \vec{A} \|_\F^2}{\alpha}.
\end{equation}
Set $\alpha = 2\varepsilon \| \vec{A} - \vec{S}\circ \vec{A} \|_\F^2$.
Then, using that $\sqrt{1+2\varepsilon} \leq 1+\varepsilon$ for all $\varepsilon>0$ and recalling \cref{eqn:AAtilde-decomp} gives that 
\begin{align*}    
\PP\bigl[\| \vec{A}-\widetilde{\vec{A}}\|_\F \ge (1+\varepsilon) \|\vec{A}-\vec{S} \circ \vec{A} \|_\F \bigr] 
&\leq \PP\bigl[\| \vec{A}-\widetilde{\vec{A}}\|_\F^2 \ge (1+2\varepsilon) \|\vec{A}-\vec{S} \circ \vec{A} \|_\F^2 \bigr] 
\\&= \PP\bigl[\|\vec{S}\circ\vec{A} - \widetilde{\vec{A}}\|_\F^2 \geq 2\varepsilon \|\vec{A} - \vec{S}\circ \vec{A}\|_\F^2 \bigr]
\\&
\leq s/\big((m-s-1)(2\varepsilon)\big). 
\end{align*}
By assumption $m \geq s(1/(2\delta\varepsilon) +1)+1$, which gives the result.
\end{proof}

We now make several comments about \cref{alg:main} and our analysis.

\begin{remark}
Our bound in~\cref{thm:ub_main_prob} has an unfavorable $O(1/\delta)$ dependence on the failure probability $\delta$.
One could apply Markov's inequality to each row and Hoeffding's inequality to the sum to obtain a dependence $O(\log(n/\delta))$. 
However this has a dependence on the dimension $n$ which we would like to avoid.
In \cref{sec:high-prob} we show that one can apply a high-dimensional analog of the ``median trick'' to obtain an algorithm with a $O(\log(1/\delta))$ failure probability (without any dependence on the dimensions $n$ and $d$).
\end{remark}

\begin{remark}
    If $\vec{A}$ and $\vec{S}$ are symmetric, then it is better to return $(\widetilde{\vec{A}} + \widetilde{\vec{A}}^\T)/2$ than $\widetilde{\vec{A}}$ since, by the triangle inequality,
\[
\| \vec{A} - (\widetilde{\vec{A}} + \widetilde{\vec{A}}^\T)/2 \|_\F
= 
\|  (\vec{A} - \widetilde{\vec{A}})/2 + (\vec{A} - \widetilde{\vec{A}})^\T/2 \|_\F
\leq \| \vec{A} - \widetilde{\vec{A}} \|_\F.
\]
\end{remark}

\begin{remark}\label{rem:diag_equiv}
If the entries of the $\vec{r}_i$ are Gaussian, then the diagonal estimator $\vec{d}_m$ from \cref{eqn:hutch_diag} is equivalent to \cref{eqn:algmain} with $\vec{S} = \vec{I}$. 
Let $\vec{r}_j$ denote the $j$th column of $\vec G$. By definition, $\vec{z}_i = [ [\vec{A}\vec{r}_1]_i, \ldots, [\vec{A}\vec{r}_m]_i]^\T$ and in this case, \[\vec{G}_i^\T := [\vec{G}]_{S_i,:} = [\vec{G}]_{i,:} = [ [\vec{r}_1]_i, \ldots, [\vec{r}_m]_i]\]
is a vector.
The $i$-th row of $\vec{d}_m$ is 
\[
[\vec{d}_m]_i 
:= \frac{\sum_{j=1}^m [\vec{r}_j]_i \cdot [\vec{A}\vec{r}_j]_i}{\sum_{j=1}^m [\vec{r}_j]_i^2}
= \frac{\vec{G}_i^\T\vec{z}_i}{\vec{G}_i^\T\vec{G}_i}
= \vec{G}_i^\dagger \vec{z}.
\]
In this sense, \cref{alg:main} for computing \cref{eqn:algmain} is a generalization of \cref{eqn:hutch_diag} to non-diagonal sparsity patterns.
Interestingly, however, we have not seen \cref{eqn:hutch_diag} interpreted in terms of a least-squares problem or pseudoinverse in the literature.
This is perhaps because past work focused on diagonal estimation (\cref{prob:recovery_sparse}) rather than approximation by a diagonal (\cref{prob:recovery_full}).
\end{remark}

\begin{remark}
\Cref{alg:main} requires solving $n$ least squares problems with a coefficient matrix of size $m\times s$. So, in addition to the application dependent cost of computing $\vec{Z} = \vec{A}\vec{G}$, its runtime is just $O(nms^2)$.
There are a number of practical improvements which can be made upon implementation.
First, for many sparsity patterns, the matrices $\vec{G}_i$ and $\vec{G}_{i+1}$ differ only by a permutation and low-rank update. 
Thus, by downdating/updating appropriate quantities, the cost of solving all $n$ least-squares problems may be lower than $n$ times the cost of solving a single system. 
In addition, a posteriori variance estimates could also be obtained through Jack-knife type techniques \cite{epperly_tropp_23}.
\end{remark}

\section{A lower-bound for adaptive algorithms}\label{sec:lower_bounds}

\Cref{alg:main} solves \cref{prob:recovery_full} using $O(s/\varepsilon)$ matvec queries. 
In this section, we show that there are distributions of matrices and sparsity patterns for which no matvec query algorithm can reliably solve \cref{prob:recovery_full} using $\Omega(s/\varepsilon)$ matvecs.
In particular, we show the following:
\lbmain*

This implies that for certain hard instances of \cref{prob:recovery_full,prob:recovery_sparse}, our \Cref{thm:ub_main_prob}  and thus \Cref{thm:ub_main} are optimal up to constants. In particular, adaptivity can only improve constants; it will not lead to an improved dependence on $s$ or $\epsilon$.
If $\vec{A}$ is known to have a particular structure, it is possible that adaptive algorithms may perform better than non-adaptive algorithms for these problems. 

We note that the condition $\varepsilon < c$ is benign. 
In particular, if $C < c/2$, then $m \leq Cs / \varepsilon$ implies $m\leq s/2$, in which case one cannot solve \cref{prob:recovery_full}, even in the zero error case, due to a parameter counting argument.

\subsection{Key technical tools}

Before we prove \cref{thm:lb_main}, we introduce several key results.

Our hard distribution will be $\vec{A} = \vec{G}^\T\vec{G}$, where $\vec{G}\sim\operatorname{Gaussian}(d)$. 
This is a special case of a so-called Wishart matrix. 
Our lower bound will make use of the fact that the conditional distribution of a Wishart matrix after a sequence of adaptive matrix-vector queries still looks like a slightly smaller transformed Wishart matrix \cite[Lemma 3.4]{braverman_hazan_simchowitz_woodworth_20}. Similar hard input distributions have been used in a number of lower-bounds for matvec query tasks \cite{simchowitz_elalaoui_recht_18,braverman_hazan_simchowitz_woodworth_20,jiang_pham_woodruff_zhang_21,chewi_dediospont_li_lu_narayanan_23}.
We believe other simple distributions such as $\vec{A}=\vec{G}$ or $\vec{A} = \vec{G} + \vec{G}^\T$ would also suffice to prove something like \cref{thm:lb_main}. 
We have chosen to use $\vec{A}= \vec{G}^\T\vec{G}$ because it is symmetric, and  it allows us to use a conceptually intuitive anti-concentration result based on the Berry--Esseen theorem.

The following is essentially Lemma 3.4 from \cite{braverman_hazan_simchowitz_woodworth_20}, restated to suit our needs:
\begin{proposition}
\label{thm:wishart}
    Suppose $\vec{G} \sim \operatorname{Gaussian}(d,r)$.
    Let $\vec{x}_1, \ldots, \vec{x}_m$ and $\vec{y}_1 = \vec{G}^\T\vec{G}\vec{x}_1, \ldots, \vec{y}_m = \vec{G}^\T\vec{G}\vec{x}_m$ be such that, for each $j=1,\ldots, m$, $\vec{x}_j$ was chosen based only on the query vectors $\vec{x}_1, \ldots, \vec{x}_{j-1}$ and the outputs $\vec{y}_1, \ldots, \vec{y}_{j-1}$.

    Then, there is an $n\times n$ orthonormal matrix $\vec{V}_m$ and an $n\times n$ matrix $\vec{\Delta}_m$, each constructed solely as functions of $\vec{x}_1, \ldots, \vec{x}_m$ and $\vec{y}_1, \ldots, \vec{y}_m$, and a matrix $\vec{G}_m \sim \operatorname{Gaussian}(d-m,r-m)$ independent of $\vec{x}_1, \ldots, \vec{x}_m$ and $\vec{y}_1, \ldots, \vec{y}_m$ such that
    \[
    \vec{V}_m^\T \vec{G}^\T\vec{G} \vec{V}_m = 
    \vec{\Delta}_m +  \begin{bmatrix}
        \vec{0}_{m,m} & \vec{0}_{m,d-m} \\
        \vec{0}_{d-m,m} & \vec{G}_m^\T \vec{G}_m
    \end{bmatrix}.
    \]
\end{proposition}

We have included a proof in \cref{sec:appendix:wishart} for completeness.

We will also use the following bound about the anti-concentration of independent random variables, which we prove in \cref{sec:anticoncentration-sums}. 
This is an immediate consequence of the Berry--Esseen Theorem and a basic  anti-concentration result for Gaussians.

\begin{proposition}\label{thm:sum-anticoncentration}
    There exists a constant $C>0$ such that, if $X_1, \ldots, X_k$ are independent random variables with $\VV[X_i] \geq \sigma^2$, and $\EE[|X_i-\EE[X_i]|^3] \leq \rho$, and if we define 
    \[ 
        X = X_1 + \cdots + X_k,
    \]
    then for any $t\in \R$ and $\alpha > 0$, if $k > C\rho^2 / (\alpha^2 \sigma^6)$, 
    \[
        \PP\Bigl[ |X - t| < \alpha \sigma \sqrt{k} \Bigr] < \alpha.
    \] 
\end{proposition}

Using \cref{thm:sum-anticoncentration}, we can derive a more specific consequence which we will use directly in the proof of \cref{thm:lb_main}.
The proof of this result is also contained in \cref{sec:anticoncentration-sums}.
\begin{lemma}\label{thm:gausian-inner-prod}
        There exists a constant $C>0$ such that, if $\vec{x} = \vec{G}\vec{u}$ and $\vec{y} = \vec{G}\vec{v}$, where $\vec{G}\sim \operatorname{Gaussian}(k)$ and $\vec{u},\vec{v}\in\R^k$ such that $\| \vec {u} \|_2 \leq 1$ and $\| \vec{v}\|_2\leq 1$,
        then for any $\alpha>0$ and $t\in\R$, if $k > C/ (\alpha^2 \| \vec{u} \|_2^6 \|\vec{v} \|_2^6)$ then
        \[
            \PP\Bigl[ \big| \vec{x}^\T\vec{y} - t \big| < \alpha \|\vec{u}\|_2\|\vec{v}\|_2 \sqrt{k} \Bigr] < \alpha.
        \]
    \end{lemma}

Finally, we will use the following observation about sparsity patterns that overlap all sufficiently large principal submatrices. 
\begin{lemma}
\label{thm:high-overlap-sparsity}
    Let $\gamma<1$ and suppose $\vec{S} \in \{0,1\}^{d\times d}$ is binary matrix for which each row and column has between $\gamma s$ and $s$ non-zero entries. 
    Let $I\subset [d]$ with $|I|\geq 2d/(2+\gamma)$. 
    Then the principal submatrix $[\vec{S}]_{I,I}$ of $\vec{S}$ contains at least $\gamma ds/(2 + \gamma)$ nonzero entries.
\end{lemma}

\begin{proof}
Let $I\subset [d]$ with $|I|\geq 2d/(2+\gamma)$. 
Note that \[
\big\| [\vec{S}]_{I,[d]} \big\|_\F^2
= \sum_{i\in I} \sum_{j\in [d]} [\vec{S}]_{i,j}
\geq |I| \cdot \gamma s
= \frac{2\gamma}{2+\gamma} ds.
\]
Next, note that, since $|I^\comp| \leq d - 2d/(2+\gamma) = \gamma d/(2+\gamma)$,
\[
    \big\| [\vec{S}]_{I,I^\comp} \big\|_\F^2
    = \sum_{i\in I} \sum_{j\in I^\comp} [\vec{S}]_{i,j}
    \leq |I^\comp| \cdot s
    \leq \frac{\gamma}{2+\gamma} ds.
\]
Finally, since $[d]$ is partitioned into $I$ and $I^\comp$, 
\[
\big\| [\vec{S}]_{I,I} \big\|_\F^2
= \big\| [\vec{S}]_{I,[d]} \big\|_\F^2
- \big\| [\vec{S}]_{I,I^\comp} \big\|_\F^2
\geq 
\frac{2\gamma}{2+\gamma} ds - \frac{\gamma}{2+\gamma}ds
= \frac{\gamma}{2+\gamma} ds. \qedhere
\]
\end{proof}

\subsection{Proof of Theorem \ref*{thm:lb_main}}

We now have the tools necessary to prove \cref{thm:lb_main}. 
The general strategy will be to show that the conditional distribution of a Wishart matrix after a sequence of (adaptive) queries is hard to approximate. 
That is, that the on-sparsity entries are anti-concentrated (conditioned on the queries) relative to the off-sparsity mass, which is $O(d^{3/2})$ with high probability. 
We will then use this result to prove \cref{thm:lb_main}.

\begin{lemma}[Main technical lemma]\label{thm:lb_lem}
Fix $\gamma\in(0,1)$.
Then there exist constants $c_1, c_2, c_3>0$ (depending only on $\gamma$) such that the following holds:

Suppose $\vec{A} = \vec{G}^\T \vec{G}$ where $\vec{G}\sim \operatorname{Gaussian}(r,d)$. 
Then, for any sparsity pattern $\vec{S}$ whose rows and columns each have between $\gamma s$ and $s$ nonzero entries, and for any (possibly randomized) algorithm that uses $m$ (possibly adaptive) matrix-vector queries to $\vec{A}$ to output $\widetilde{\vec{A}}$ with $\widetilde{\vec{A}} \circ \vec S = \widetilde{\vec{A}}$, if $m\leq c_1 d$ and, for any $\alpha \in (0,1)$, $r > m + c_3 / \alpha^2$, then
\[
    \PP\Big[ \| \vec{S}\circ \vec{A} -\widetilde{\vec{A}} \|_\F^2
    < c_2 \alpha^2 ds (r-m)
    \Big] < \alpha.
\]
\end{lemma}

\begin{proof}
    Fix $\gamma \in (0,1)$.
    Let $C$ denote the absolute constant in \cref{thm:gausian-inner-prod}, define 
    \[
        c_1 = \frac{\gamma}{4+2\gamma}
        \qquad
        c_2 = \frac{c_1^2}{4}
        \qquad
        c_3 =\frac{C}{c_1^6},
    \]
    and suppose $m$, $r$, and $d$ are integers such that
    \[ 
    m \leq c_1 d
    ,\qquad
    r - m > \frac{c_3}{\alpha^2}.
    \]

    Suppose we do $m$ (possibly adaptive) queries to $\vec{A}$.
    \Cref{thm:wishart} implies that there exists an $d\times d$ matrix $\vec{\Delta}$, and $d\times (d-m)$ matrix $\vec{V}$ with orthonormal columns, both constructed solely as functions of the queries and measurements, and a matrix $\widehat{\vec{G}}\sim\operatorname{Gaussian}(r-m,d-m)$ independent of the queries and measurements such that
    \[
    \vec{A} 
    = \vec{\Delta} + \vec{V} \vec{W} \vec{V}^\T
    ,\qquad
    \vec{W} = \widehat{\vec{G}}^\T \widehat{\vec{G}}.
    \]
    Let $\vec{S}$ be any sparsity pattern for which each row and column has between $\gamma s$ and $s$ non-zeros; i.e. for which we can apply \cref{thm:high-overlap-sparsity}.

    Let $\vec{T}\in \R^{d\times d}$ be any matrix with sparsity $\vec{S}$ determined solely as a function of the queries and measurements, and hence independent of $\vec{G}$.
    Without loss of generality, we will absorb $\vec{S}\circ\vec{\Delta}$ into $\vec{T}$. 
    Note also that it suffices to assume $\vec{T}$ is  deterministic, as the following argument holds for all possible draws of a random $\vec{T}$ (and by extension, for the expectation over random draws of $\vec{T}$).
    
    Define the set of indices
    \[
        P =  \Big\{ (i,j) \in [d]\times [d] : \big| [\vec S \circ \vec{V}\vec{W}\vec{V}^\T]_{i,j} - [\vec{T}]_{i,j} \big|^2 >  (\alpha/2)^2 c_1 (r-m) \Big\},
    \]
    and the event
    \[
        E =  \Big\{ |P| \geq c_1 ds  \Big\}.
    \]
    Note that if $E$ holds, we have that 
    \begin{align*}
        \| \vec{S}\circ \vec{A} - \vec{T} \|_\F^2
        &\geq \sum_{(i,j)\in P} \big| [\vec S \circ \vec{V}\vec{W}\vec{V}^\T]_{i,j} - [\vec{T}]_{i,j} \big|^2
        \\&\geq |P| \cdot (\alpha/2)^2 c_1 (r-m)
        \geq c_1 ds \cdot (\alpha/2)^2 c_1 (r-m)
        = \alpha^2 c_2 ds(r-m),
    \end{align*}    
    so it remains to show $\PP[E] \geq 1-\alpha$.
    
    Towards this end, let $\vec{v}_i$ be the $i$th row of $\vec{V}$, and define
    \[
    I = \big\{ i \in [d] :  \| \vec{v}_i \|_2^2 \geq c_1 \big\}
    ,\qquad 
    M = \big\{ (i,j) \in I\times I :  [\vec{S}]_{i,j} = 1 \big\}. 
    \]    
    We must have $|I| \ge 2d/(2+\gamma)$. 
    Otherwise, since $\vec V$ has orthonormal columns so that $\|\vec{v}_i\|_2^2\leq 1$, we would have:
    \[
    d-m = \|\vec{V}\|_\F^2 
    = \sum_{i=1}^{d} \| \vec{v}_i \|_2^2
    < \sum_{i\in I^\comp} c_1 + \sum_{i\in I} 1
    < c_1 d + \frac{2d}{2+\gamma} 
    = (1-c_1) d,
    \]
    which contradicts our assumption $m \leq c_1 d \Leftrightarrow d-m \geq (1-c_1)d$. 
    Then, since $|I| \geq 2d/(2+\gamma)$, \cref{thm:high-overlap-sparsity} and our assumption on $\vec{S}$ give that 
    \[
    |M| \geq \frac{\gamma ds}{2 + \gamma} = 2c_1 ds.
    \]

    We now show that if $(i,j)\in M$ then $(i,j)\in P$ with high probability.
    Fix arbitrary $(i,j)\in M$ (and note that this means $i,j \in I$).
    Since $[\vec{S}]_{i,j} = 1$, note that $[\vec S \circ \vec{V}\vec{W}\vec{V}^\T]_{i,j} = \vec{x}_i^\T\vec{x}_j$, where $\vec x_i = \widehat{\vec{G}} \vec v_i$ and $\vec x_j = \widehat{\vec{G}} \vec v_j$ and $\vec{v}_i$ and $\vec{v}_j$ are the $i$-th and $j$-th rows of $\vec{V}$ respectively.  
    We also have that $\|\vec{v}_i\|_2^2, \|\vec{v}_j\|_2^2 \geq c_1$.
    With this in mind, our choice of constants implies 
    \[
        r-m 
        \geq \frac{c_3}{\alpha^2 } 
        = \frac{C}{\alpha^2c_1^6}
        \geq \frac{C}{\alpha^2 \| \vec{v}_i \|_2^6 \|\vec{v}_j\|_2^6}.
    \]
    Hence, applying \cref{thm:gausian-inner-prod},
    \begin{align*}
        \PP\Bigl[(i,j) \not\in P \Bigr]
        &= \PP\Bigl[  \big| [\vec S \circ \vec{V}\vec{W}\vec{V}^\T]_{i,j} - [\vec{T}]_{i,j} \big|  < (\alpha/2) c_1 \sqrt{r-m} \Bigr] 
        \\&\leq 
        \PP\Bigl[  \big| \vec{x}_i^\T\vec{x}_j - [\vec{T}]_{i,j} \big|  < (\alpha/2) \|\vec{v}_i\|_2 \| \vec{v}_j \|_2 \sqrt{r-m} \Bigr] 
        < \frac{\alpha}{2}.
    \end{align*}
    Now, by Markov's inequality, we have that
    \[
        \PP\Biggl[ \sum_{(i,j)\in M} \ones[(i,j)\not\in P ] 
        \geq \frac{1}{\alpha} \cdot \frac{\alpha}{2} |M| \Biggr] 
        \leq \alpha.
    \]
    Since $|M| \geq 2c_1 ds$, this then implies that 
    \[
        \PP\bigl[E\bigr]
        = \PP\biggl[ |P| \geq c_1 ds \biggr] 
        \geq \PP\biggl[ |P| \geq\frac{|M|}{2} \biggr] 
        \geq 1-\alpha.
    \]
    This proves the result.
\end{proof}

With \cref{thm:lb_lem}, the proof if \cref{thm:lb_main} is straightforward. 
The basic idea is that if $r = d$ and $\alpha$ is constant then $\|\vec{S}\circ\vec{A} - \widetilde{\vec{A}}\|_\F^2$ is typically $\Omega(sd^2)$. 
However, by simple computation, it is not hard to see that $\|\vec{S}\circ \vec{A}\|_\F^2 \leq \| \vec{A} \|_\F^2$ is typically $O(d^3)$. 
Thus, if we set $d = O(s/\varepsilon)$, we typically will have that $\| \vec{S} \circ \vec{A} - \widetilde{\vec{A}}\|_\F^2 > \varepsilon \| \vec{S} \circ \vec{A} \|$ as long as $m\leq c_1d = O(s/\varepsilon)$.

\begin{proof}[Proof of \cref{thm:lb_main}]
    Fix $\gamma \in (0,1)$ and let $c_1, c_2, c_3$ be the constants from \cref{thm:lb_lem}.
    We will make the following assignments:\footnote{We have labeled the constants so that if $c_i$ depends on $c_j$, then $i > j$.}
    \begin{align*}
    \alpha &= \frac{1}{25}
    ,&
    b_1 &= 50
    ,&
    b_2 &= \frac{c_2\alpha^2 (1-c_1)}{6 b_1}
    ,&     
    b_3 &= \max\left\{ 2, \frac{(1-c_1)b_2\alpha^2}{c_3} \right\}
    ,&
    C_1 &= c_1 b_2.
    \end{align*}  
    Fix $\varepsilon > 0$. We will show that if 
    \begin{equation}
    \label{eqn:lb_main:assm}
    m\leq C_1 \frac{s}{\varepsilon}
    , \qquad
    \frac{b_2}{\varepsilon}\in\mathbb{Z}
    ,\qquad
    \varepsilon < \frac{1}{b_3} \leq \frac{1}{2},
    \end{equation}
    then there is a distribution on matrices such that for any sparsity pattern with between $\gamma s$ and $s$ entries per row,
    \begin{equation}\label{eqn:lb_main:res}
    \PP\Bigl[ \| \vec{A} - \widetilde{\vec{A}} \|_\F
    \leq (1+2\varepsilon) \,\| \vec{A} - \vec{S}\circ \vec{A} \|_\F \Bigr] \leq \frac{1}{25}.
    \end{equation}
    The assumption $b_2/\varepsilon\in\mathbb{Z}$ will subsequently be removed.

    Towards this end, assume \cref{eqn:lb_main:assm}, and set
    \[
        d = b_2 \frac{s}{\varepsilon}
        ,\quad 
        r = d
        ,\quad
        \vec{A} = \vec{G}^\T \vec{G}
        ,\quad
        \vec{G} \sim \operatorname{Gaussian}(r,d).
    \]
    Define events,
    \[
        E = \Big\{   \| \vec{S}\circ \vec{A} -\widetilde{\vec{A}} \|_\F^2 \geq c_2 \alpha^2 ds (r-m) \Big\}
        ,\qquad
        F =  \Big\{ \| \vec{A} - \vec{S}\circ\vec{A} \|_\F^2 \leq b_1 d^{3} \Big\}.
    \]
    By assumption, $m\leq C_1 s / \varepsilon = c_1 d$ and 
    \[
    r - m \geq (1-c_1) d 
    = (1-c_1) \cdot b_2 \frac{s}{\varepsilon}
    > (1-c_1) b_2 \frac{1}{b_3}
    = \frac{c_3}{\alpha^2}.
    \]
    Therefore, applying \cref{thm:lb_lem}, we have that $\PP[E] \geq 49/50$.
    
    It is easy to show that $\EE[\|\vec{A} - \vec{S}\circ\vec{A} \|_\F^2] \leq d^3$ (see \cref{thm:Aexpectednorm} for a derivation), so since $b_1 = 50$, an application of Markov's inequality implies $\PP[F] \geq 49/50$.

    Note that 
    \[
    6\varepsilon b_1 d^3
    \leq 6\varepsilon b_1 d \cdot b_2\frac{s}{\varepsilon}\cdot \frac{r-m}{1-c_1}
    \leq
    c_2 \alpha^2 ds (r-m).
    \]
    By a union bound, both $E$ and $F$ hold with probability at least $24/25$.
    Then since $\vec{S}\circ\vec{A}$ and $\widetilde{\vec{A}}$ have disjoint support and  $\varepsilon < 1/2$, as noted in \cref{eqn:AAtilde-decomp},
    \[
    \| \vec{S}\circ \vec{A} - \widetilde{\vec{A}} \|_\F^2
    \geq (1+6\varepsilon) \| \vec{A}-\vec{S}\circ\vec{A}\|_\F^2
    \geq (1+2\varepsilon)^2 \| \vec{A}-\vec{S}\circ\vec{A}\|_\F^2.
    \]
    
    We now relax the assumption $c_5/\varepsilon\in\mathbb{Z}$.
    For arbitrary $\varepsilon>0$, define $\tilde{\varepsilon} = b_2 / \lceil b_2 / \varepsilon \rceil$.
    Then $b_2 / \tilde{\varepsilon} \in\mathbb{Z}$ and $\tilde{\varepsilon} \leq \varepsilon$.
    Moreover, if $\varepsilon \leq b_2 / 2$, then $\varepsilon \leq  b_2 \tilde{\varepsilon} / (b_2 - \tilde{\varepsilon}) \leq 2\tilde{\varepsilon}$.
    Set $c = \min\{1/b_3, b_2/2 \}$.
    Therefore, since $1/\varepsilon\leq 1/\tilde{\varepsilon}$, if 
    \begin{equation}
    m\leq C_1 \frac{s}{\varepsilon}
    ,\qquad
    \varepsilon < c,
    \end{equation}
    then the conditions in \cref{eqn:lb_main:assm} hold (with $\tilde{\varepsilon}$), and so we have the result in \cref{eqn:lb_main:res} (with $\tilde{\varepsilon}$).
    Since $\varepsilon \leq 2\tilde{\varepsilon}$, this implies there is a distribution on matrices and sparsity patterns for which the result of the theorem holds.

    Thus, the proof is complete with $C = 2C_1$.
\end{proof}

\section{Comparison with coloring methods}\label{sec:coloring}

A number of methods for sparse-matrix and operator recovery based on \emph{graph colorings} have been proposed \cite[etc.]{curtis_powell_reid_74,coleman_more_83,stathopoulos_laeuchli_orginos_13,schafer_owhadi_21,frommer_schimmel_schweitzer_21}.\footnote{These methods are sometimes called probing methods in the matrix function trace estimation literature.}
To the best of our knowledge, such methods were first considered in the 1970s, and are based on the following observation:
Partition column indices $[d]$ into sets $\{C_1, \ldots, C_k\}$ such that the columns of $\vec{A}$ corresponding to any given $C_i$ have disjoint support.
Then we can recover all of the columns in a given set $C_i$ with a single matrix-vector product: a vector which is supported only on $C_i$.
The sets $C_i$ can be obtained by coloring a graph.
In particular, form a graph on $d$ vertices, where there is an edge between vertices $i$ and $j$ if and only if the $i$-th and $j$-th columns of $\vec{A}$ have overlapping support.
A $k$-coloring of this graph gives the partition $\{C_1, \ldots, C_k\}$.

For matrices which cannot be colored with a small number of colors, one might still use coloring methods to partition the columns of the matrix, and then recover the relevant entries within each partition using an algorithm which can handle noise (e.g. \cref{alg:main} or Hutchinson's diagonal estimator).
In \cref{sec:banded} we analyze this approach for banded matrices. 
The remainder of this section provides some extreme cases to illustrate some potential pros and cons of coloring-based methods in comparison to our proposed method.

\subsection{Analysis of coloring methods on banded matrices}\label{sec:banded}

We will describe how the estimator $\vec{d}_m$ defined in \cref{eqn:hutch_diag} can be combined with coloring methods to solve \cref{prob:recovery_full,prob:recovery_sparse}.
Note that if the entries of $\vec{r}_i$ are chosen as independent Rademacher random variables (i.e. each entry is independently $+1$ with probability $1/2$ and $-1$ with probability $1/2$), then \cref{eqn:hutch_diag} simplifies, as $\vec{r}_i\circ\vec{r}_i$ is always the all-ones vector.
It is then easy to show $\EE[\vec{d}_m] = \diag(\vec{A})$ and $\EE[\|\diag(\vec{A}) - \vec{d}_m\|_2^2] = \| \vec{A} - \vec{I}\circ\vec{A}\|_\F^2 / m$.

For any integer $b\geq 0$, let $\vec{S}\in \{0,1\}^{d\times d}$ be the sparsity pattern of banded matrices of bandwidth $s = 2b+1$; that is, $[\vec{S}]_{i,j} = \ones(|i-j|\leq b)$.
For convenience, we will assume $d = ks$, for some  integer $k \geq 1$.
This sparsity pattern yields a natural coloring-based partitioning of $[d]$ into the sets $C_i = \{i + js : j\in[k] \}$ for $i \in [s]$.

For each $i \in [s]$, define the $d\times k$ matrices
\[
\vec{A}^{(i)} = [\vec{A}]_{:,C_i}
,\qquad
\vec{S}^{(i)} = [\vec{S}]_{:,C_i}.
\]
Observe that if $\vec{A} = \vec{A}\circ\vec{S}$, then we could recover all of the entries in $\vec{A}^{(i)}$ by multiplying with the all-ones vector, since there would be exactly one nonzero in each row of $\vec{A}^{(i)}$.

Let $\vec{v}^{(i)} \in \{-1,+1\}^k$ have independent Rademacher entries
Define $\vec{c}^{(i)} = \vec{S}^{(i)} \vec{v}^{(i)}$ and consider the vector 
\[
\vec{y}^{(i)} = \vec{c}^{(i)} \circ (\vec{A}^{(i)} \vec{v}^{(i)}).
\]
Since the entries of $\vec{v}^{(i)}$ are independent, mean zero, and variance 1, a direct computation shows that 
\[
\EE\bigl[ \vec{y}^{(i)} \bigr]
= (\vec{S}^{(i)}\circ \vec{A}^{(i)}) \vec{1}
,\quad
\EE\bigl[ \| \vec{y}^{(i)} - (\vec{S}^{(i)}\circ \vec{A}^{(i)}) \vec{1} \|_2^2 \bigr]
= 
\| \vec{A}^{(i)} - \vec{S}^{(i)} \circ \vec{A}^{(i)} \|_\F^2.
\]
Matrix-vector products with $\vec{A}^{(i)}$ can be computed with a single product to $\vec{A}$.
Thus, using $s$ matrix-vector products, we obtain an unbiased estimator for $\vec{S}\circ \vec{A}$ with expected squared error $\|\vec{A} - \vec{S}\circ\vec{A}\|_\F^2$.
Averaging $t$ independent copies of this estimator will reduce the variance by a factor of $t$.
Hence, using $m \geq s/\varepsilon$ matrix-vector products, one obtains an algorithm with expected squared error bounded by $\varepsilon \|\vec{A} - \vec{S}\circ\vec{A}\|_\F^2$.
While this algorithm is well-known in the literature \cite[etc.]{stathopoulos_laeuchli_orginos_13,frommer_schimmel_schweitzer_21}, to the best of our knowledge, an analysis like the one described here has not been written down. 
As we discuss in the next section, this is marginally better than \cref{thm:ub_main}.

\subsection{Coloring algorithms can be better}\label{sec:coloring_better}

The previous example shows that if $\vec{S}$ is a banded matrix, the expected squared error of the coloring-based method  described above after $m$ matrix-vector products is  better by a factor of $(m-s-1)/m$ than \cref{alg:main}. 
If $s$ is large, $m$ is not much larger than $s$, and the off-sparsity mass is large, this may be relevant.
However, when $m$ is large relative to $s$ or if the off-sparsity mass is small, this difference is not so important.

Coloring based methods can also outperform \cref{alg:main} because the error of coloring methods decouples entirely between colors. 
Thus, to recover the entries within a given color, algorithms do not need to pay for large off-sparsity entries in a different color.
An example matrix for which this observation leads to an arbitrarily large improvement over \cref{alg:main} is shown in the left panel of \cref{fig:sparsity-vis}.
One can also use extra queries to reduce the variance in only colors with large off-sparsity mass.

Such an improvement is clearly tied to how much of the off-sparsity mass can be ignored by the coloring scheme.
If this mass is small, then coloring-based methods unnecessarily use extra queries for each color.
For instance, the middle panel of \cref{fig:sparsity-vis} shows an example where coloring would not be so beneficial.

\subsection{Coloring algorithms can be worse}
\label{sec:hard-coloring}

In the zero-error case, it is clear that an $s$-sparse matrix requires at least $s$ colors (and hence matvecs) to recover $\vec{A}$.
As noted in \cref{thm:exact_sparse}, \cref{alg:main} requires exactly $s$ matvecs in the zero-error case, matching this lower bound. However, coloring-based methods  can fail to match this lower bound, under-performing \cref{alg:main}.

In particular, for any $s,d\geq 1$, there are $s$-sparse matrices with $d$ columns which do not have a column partitioning into fewer than $d$ partitions; i.e. for which every pair of columns has intersecting support and hence coloring-based approach do no better than the trivial algorithm which reads the matrix column-by-column, using $d$ matvecs. 
A natural assumption, motivated by the banded case, is that the matrix $\vec{A}$ is $s$-doubly sparse. 
That is, there are at most $s$ nonzeros in any given row or column.
For an $s$-doubly sparse matrix, the maximum vertex degree of the graph described above is trivially bounded by $s^2$, so a greedy coloring will result in at most $s^2+1$ colors.
The following example shows there are $s$-doubly sparse matrices for which $\Omega(s^2)$ colors are required.
Thus, while coloring based methods would require $\Omega(s^2)$ matvec queries to recover such a matrix, \cref{alg:main} requires only $s$ queries.

\begin{figure}
    \centering
    \includegraphics[scale=.4]{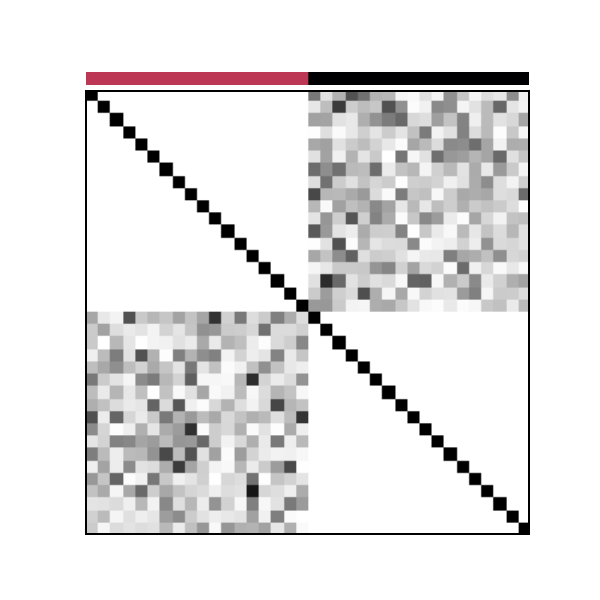}
    \hfil
    \includegraphics[scale=.4]{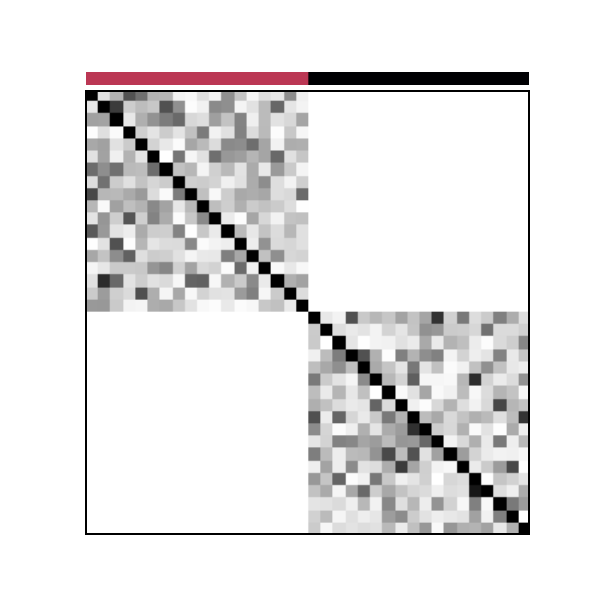}    
    \hfil
    \includegraphics[scale=.4]{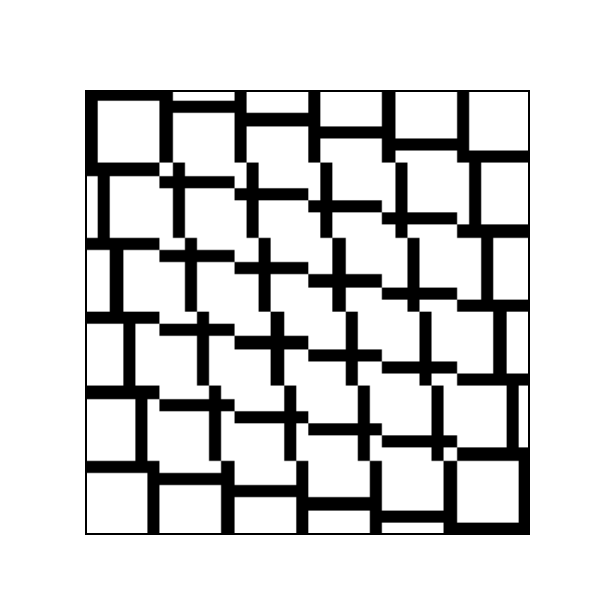}
    \caption{
    \emph{Left}: Visualization of a matrix describedin \cref{sec:coloring_better} for which \cref{alg:main} is not the best method for recovering the diagonal (intensity indicates magnitude of entries of $\vec{A}$).
    In particular, the diagonal of the matrix can be recovered using exactly 2 queries, while \cref{alg:main} will require many queries to overcome the large noise in the off-diagonal blocks. 
    \emph{Middle}: Visualization of a matrix for which using the same colorings as the matrix on the left panel will not help.
    \emph{Right}: Visualization of the hard sparsity pattern described in \cref{sec:hard-coloring} with $k=10$. 
    Here black pixels correspond to one and white pixels to zero.
    Note that while each row and column of the matrix has only $O(k)$ nonzeros, each pair of the $k^2$ columns has overlapping support.
    }
    \label{fig:sparsity-vis}
\end{figure}

In particular, for any integer $k\geq 1$, define the $k^2\times k^2$ matrix $\vec{A}$ by
\[
[\vec{A}]_{pk+i,qk+j}
= \begin{cases}
    1 & i=q \text{ or } j=p\\ 
    0 & \text{otherwise}
\end{cases}
,\qquad
p,q,i,j \in [k] .
\]
This sparsity pattern is represented in the left panel of \cref{fig:sparsity-vis}.
Any given row or column of this matrix has exactly $s = 2k-1$ nonzeros. 
However, the support of every column overlaps. 
Indeed, for columns $x = pk + i$ and $y=qk+j$ (with $i,j\in[k]$),
\[
[\vec{A}]_{ik+q,x} = [\vec{A}]_{jk+p,y} = 1
,\quad
[\vec{A}]_{jk+p,x} = [\vec{A}]_{ik+q,y} = 1.
\]
Therefore, exact coloring based approaches require $d = k^2 = O(s^2)$ colors; i.e. they do no better than the trivial upper bound of $d$ queries.
In contrast, \cref{alg:main} would require only $s$ queries.

\section{Numerical Experiments}\label{sec:numerical}

In this section we provide several numerical experiments which illustrate the performance of \cref{alg:main}.
These problems are modeled after similar problems from the literature.
Code to reproduce the figures can be found at 
\url{https://github.com/tchen-research/fixed_sparsity_matrix_approximation}.

\subsection{Model problem}

We consider the matrix $\vec{A} = \vec{M}^{-1}$, where $\vec{M} = \operatorname{tridiag}(-1,4,-1)$.
This class of matrices exhibits exponential decay away from the diagonal and was used in experiments in past work \cite{benzi_simoncini_15,frommer_schimmel_schweitzer_21}.

We take $\vec{A}$ to be $1000\times 1000$ and, for varying values of $b\geq 0$, we set $\vec{S}$ to be a symmetric banded matrix of maximum total bandwidth $2b+1$; i.e. $[\vec{S}]_{i,j} = \ones(|i-j|\leq b)$.
We then compute the approximation error $\|\vec{A} - \widetilde{\vec{A}}\|_\F$ and recovery error $\|\vec{S}\circ \vec{A} - \widetilde{\vec{A}}\|_\F$ for the output of \cref{alg:main} run using a varying number of matvec queries $m$.

The results are illustrated in \cref{fig:inverse_decay}.
Here the convergence of \cref{alg:main} is matched well by the upper bound in \cref{thm:ub_main} for the expected squared error (note that the plot shows the error, not the expected squared error).

\begin{figure}[h]
    \centering
    \includegraphics[scale=.5]{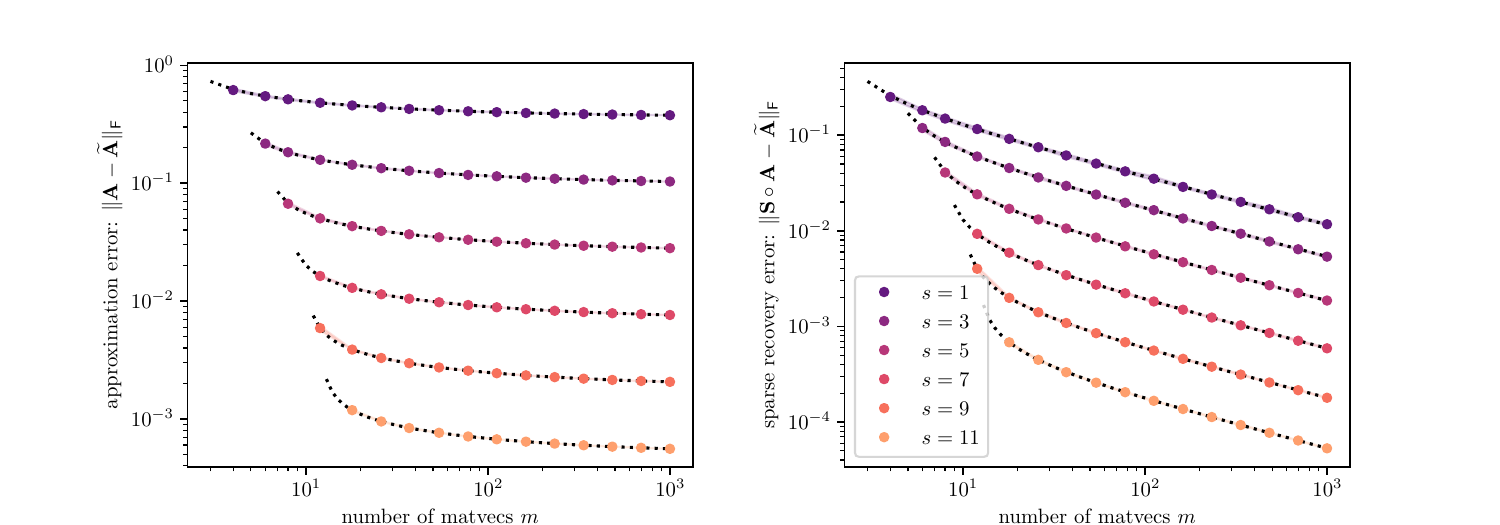}
    \caption{Approximation of model problem matrix $\vec{A} = \operatorname{tridiag}(-1,4,-1)^{-1}$ by a matrix of total bandwidth $s$ for varying values of $s$.
    The solid circles indicate the root mean squared error of \cref{alg:main} over 20 independent runs of the algorithm, and the shaded region indicates the 10\%-90\% range. 
    The dotted lines are the $\sqrt{s/(m-s-1)}\|\vec{A} -\vec{S}\circ\widetilde{\vec{A}}\|_\F$ (left) and $\sqrt{1+s/(m-s-1)}\|\vec{A} -\vec{S}\circ\widetilde{\vec{A}}\|_\F$ (right).
    }
    \label{fig:inverse_decay}
\end{figure}

\subsection{Trefethen Primes}

We let $\vec{A} = \vec{M}^{-1}$, where $\vec{M}$ be the $1000 \times 1000$ matrix whose entries are zero everywhere except for the primes $2, 3, 5, 7, \ldots, 7919$ along the main diagonal and the number 1 in all the positions $[\vec{B}]_{i,j}$ with $|i – j| \in \{1,2,4,8, \ldots,512\}$.\footnote{In problem 7 of the ``A Hundred-dollar, Hundred-digit Challenge'' in SIAM News, readers are asked to compute the (1,1) entry of a larger but analogously defined matrix $\vec{A}$ to 100 digits of accuracy \cite{trefethen_02}.}
An example (somewhat different from our example) involving this matrix was used in \cite{park_nakatsukasa_23}.

\begin{figure}[h]\centering
    \includegraphics[scale=.4]{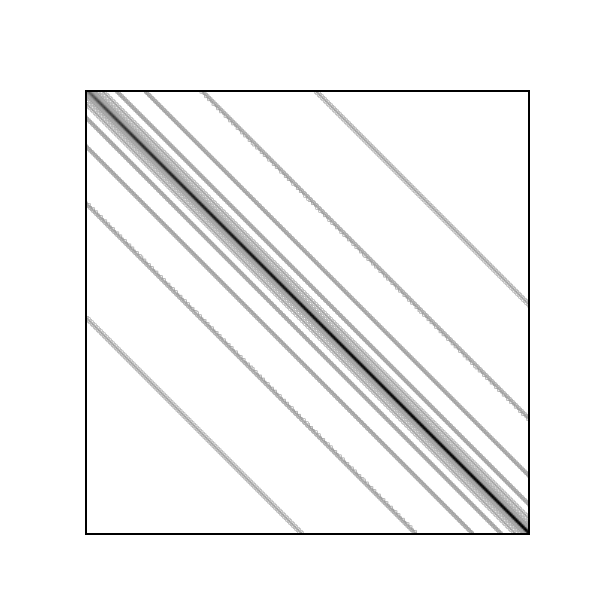}
    \hfil
    \includegraphics[scale=.4]{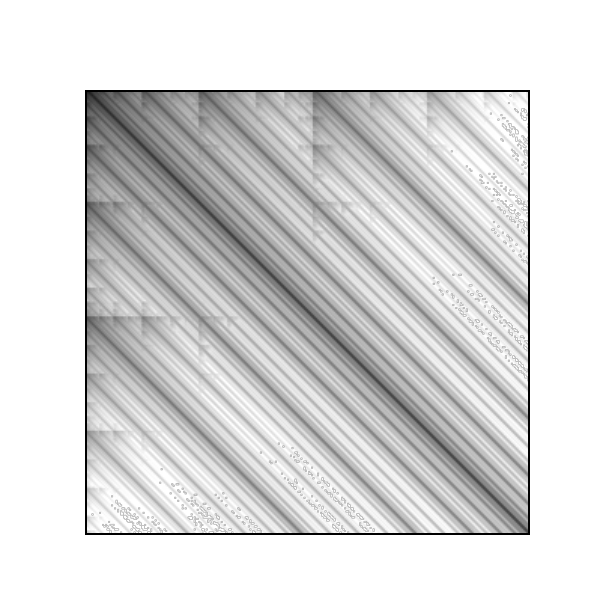}
    \hfil
    \includegraphics[scale=.4]{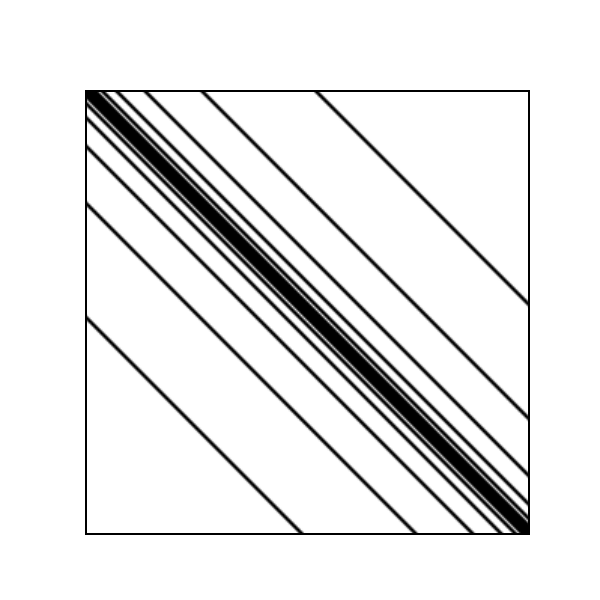}
    \caption{
    \emph{Left}: Log-scale of the nonzero entries of $\vec{M}$, which range in magnitude from 1 to 7919.
    \emph{Middle}: Log-scale of the nonzero entries of $\vec{A}$.
    \emph{Right}: Sample sparsity pattern $\vec{S}$ corresponding to $b=5$.
    }
    \label{fig:trefethen_sparsity-vis}
\end{figure}

For $b>0$, we defined a sparsity pattern $\vec{S}$ to be such that, for each $t\in \{1,2,4,8,\ldots, 512\}$, $[\vec{S}]_{i,j} = 1$ whenever $|i-j\pm t|\leq b$ and zero otherwise. 
In other words, the sparsity pattern consists of bandwidth $2b+1$ bands centered nonzero entries of $\vec{M} = \vec{A}^{-1}$.
This is illustrated in \cref{fig:trefethen_sparsity-vis}.

The results are shown in \cref{fig:trefethen_inv}.
Here, the convergence of \cref{alg:main} is somewhat better than the upper bound \cref{thm:ub_main}.
This is because many rows of the sparsity pattern have far fewer than $s$ entries.
From the proof of \cref{thm:ub_main}, it is clear how to obtain an exact characterization of the expected squared error.

\begin{figure}[h]
    \centering
    \includegraphics[scale=.5]{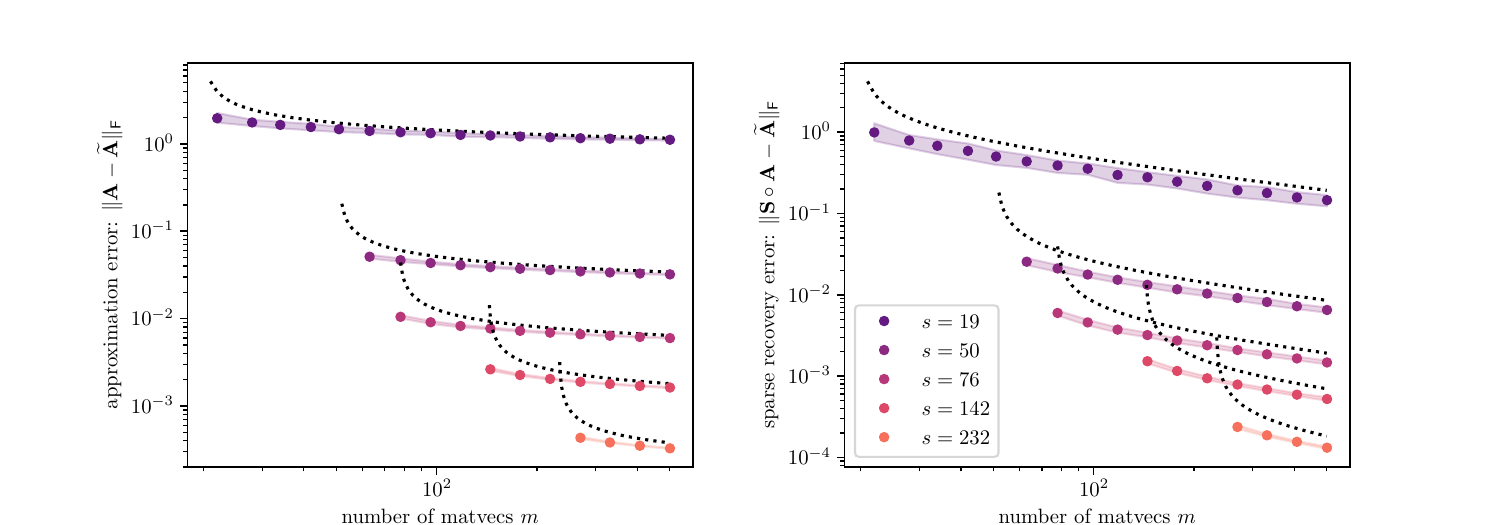}
    \caption{Approximation of ``Trefethen primes'' inverse matrix by a multi-banded matrix for varying values of $s$.
    The solid circles indicate the root mean squared error of \cref{alg:main} over 100 independent runs of the algorithm, and the shaded region indicates the 10\%-90\% range. 
    The dotted lines are the $\sqrt{s/(m-s-1)}\|\vec{A} -\vec{S}\circ\widetilde{\vec{A}}\|_\F$ (left) and $\sqrt{1+s/(m-s-1)}\|\vec{A} -\vec{S}\circ\widetilde{\vec{A}}\|_\F$ (right).
    }
    \label{fig:trefethen_inv}
\end{figure}

\section{Outlook}
\label{sec:outlook}

This work raises a number of interesting practical and theoretical questions. We now comment on three.

It is clear that there is potential for combining coloring methods with algorithms such as \cref{alg:main} in order to avoid paying for large off-sparsity entries in parts of a matrix while simultaneously maintaining the robustness to noise enjoyed by \cref{alg:main}.
One approach to combining these two paradigms is to use adaptive queries to identify portions of the matrix with large-mass, and then find a coloring based on this information.
We believe further study in this direction may yield algorithms which work better in many practical situations.
Of course, our lower bound \cref{thm:lb_main} shows that only constant factors can be improved for some families of problem instances.

The present paper focuses only on the Frobenius norm.
It would be valuable to understand the analogous problem in other norms such as the matrix 2-norm.
For other norms, $\vec{S}\circ\vec{A}$ is not necessarily the best approximation to $\vec{A}$ with sparsity $\vec{S}$.
For instance, if $\vec{A} = [1,1;1,1]$ and $\vec{S} = [1,0;0,0]$, then 
\[
\operatornamewithlimits{argmin}_{\vec{X}=\vec{S}\circ\vec{X}} \|\vec{A} - \vec{X}\|_2
= \begin{bmatrix}
2 & 0 \\
0 & 0 
\end{bmatrix}
\neq 
\begin{bmatrix}
1 & 0 \\
0 & 0 
\end{bmatrix}
= \vec{S}\circ\vec{A}.
\]
Loosely speaking, minimizing the operator-norm approximation error requires the rows of the error matrix to be small and unaligned, whereas the Frobenius norm problem only requires them to be small.

Finally, it is of broad interest to understand algorithms and lower bounds for richer structures of matrices, such as the sum of sparse and low-rank matrices and matrices with hierarchical low-rank structure. This paper is a starting point for investigating these problems, and we hope that future work will explore them in greater depth.

\section*{Acknowledgements}

Diana Halikias was supported by the Office of Naval Research (ONR) under grant N00014-23-1-2729. 
Cameron Musco was partially supported by NSF Grant 2046235.
Christopher Musco was partially supported by NSF Grant 2045590.

We thank Robert J. Webber for informing us of the operator-norm example in \cref{sec:outlook}.

\appendix

\section{Discussion on relative error recovery}
\label{sec:relative_discussion}

Instead of \cref{prob:recovery_sparse}, one may hope for a guarantee like
\begin{equation}\label{prob:rel}
\| \vec{S} \circ \vec{A} - \widetilde{\vec{A}} \|_\F
\leq \varepsilon \| \vec{S} \circ \vec{A} \|_\F.
\end{equation}    
Unfortunately, we cannot expect to be able to obtain such an approximation using fewer than $d$ queries in general; for $\varepsilon<1$, this would require differentiating matrices for which $\vec{S}\circ\vec{A}$ is zero from those for which it is arbitrarily small.

Even so, in many situations where one might hope to recover the sparse part of a matrix, $\| \vec{A} - \vec{S} \circ \vec{A} \|_\F$ can be bounded by a constant multiple of $\| \vec{S} \circ \vec{A} \|_\F$.
In such cases, \cref{thm:ub_main} implies that \cref{alg:main} returns an approximation satisfying \cref{prob:rel} using $O(s/\varepsilon^2)$ queries.
One instance where this can be expected is if $\vec{A} = f(\vec{M})$, where $\vec{M}$ has sparsity $\vec{S}$. 
It is known that matrix functions of sparse matrices have entries which decay exponentially away from the sparsity pattern \cite{benzi_razouk_08}.
Another instance in which this holds is if $\vec{S}$ contains the identity and $\vec{A}$ is diagonally dominant; i.e. if $|[\vec{A}]_{i,i}| \geq \sum_{j\neq i} |[\vec{A}]_{i,j}|$. 
Indeed, in this case
\[
\| \vec{A} - \vec{S} \circ \vec{A} \|_\F^2
\leq \| \vec{A} - \vec{I} \circ \vec{A} \|_\F^2
= \sum_{i=1}^{d} \sum_{\substack{j=1\\j\neq i}}^{d} [\vec{A}]_{i,j}^2
\leq \sum_{i=1}^{d} [\vec{A}]_{i,i}^2
= \| \vec{I} \circ \vec{A} \|_\F^2
\leq \| \vec{S}\circ \vec{A}\|_\F^2.
\]

Using \cref{thm:lb_lem}, we can derive lower-bounds. 
For instance, if in the proof of \cref{thm:lb_main} we replace $F$ by \[
    F =  \Big\{ \| \vec{I}\circ\vec{A} \|_\F^2 \leq b_1 d^{3} \Big\}
\]
and $b_1$ by $150$, then by \cref{thm:Aexpectednorm}, the proof of \cref{thm:lb_main} can be repeated and results in a bound $\Omega(1/\varepsilon^2)$ queries for relative error diagonal approximation of PSD matrices.

In fact, if we set $r$ as large constant multiple of $d$, then $\vec{A}$ becomes diagonally dominant with high probability, and the same lower bound still holds. 
In light of the above upper bound, \cref{alg:main} is optimal for relative error diagonal approximation on diagonally dominant matrices.

\section{Lower Bound Lemmas}
\label{sec:lower_bound_lemmas}

In this section, we provide the proofs of the key lemmas used the prove the lower bounds stated in \cref{sec:lower_bounds}.
These lemmas are all standard in the literature, but we include the proofs for the benefit of the reader, as they are, for the most part, self-contained and interesting.

\subsection{Adaptive queries to Wishart matrices}
\label{sec:appendix:wishart}

In this section we provide a proof of \cref{thm:wishart}, which is essentially \cite[Lemma 3.4]{braverman_hazan_simchowitz_woodworth_20}.
This is mostly included for completeness.
First, we consider what happens after a single non-adaptive query.

\begin{lemma}\label{thm:wishart_1step}
Suppose $\vec{G} \sim \operatorname{Gaussian}(d,r)$.
Let $\vec{x}$ be a unit-length query chosen independently of $\vec{G}$, and define $\vec{y} = \vec{G}^\T\vec{G} \vec{x}$.
Then, there is an $n\times n$ orthonormal matrix $\hat{\vec{X}} = [ \vec{x} ~ \vec{X}]$, constructed solely as a function of $\vec{x}$, and a matrix $\vec{H} \sim \operatorname{Gaussian}(d-1,r-1)$ independent of $\vec{x}$ and $\vec{y}$ such that such that 
\[
\hat{\vec{X}}^\T \vec{G}^\T\vec{G} \hat{\vec{X}} = 
\begin{bmatrix}
    \vec{x}^\T \vec{y} & \vec{y}^\T\vec{X} \\
    \vec{X}^\T \vec{y} & (\vec{x}^\T\vec{y})^{-2}\vec{X}^\T\vec{y}\vec{y}^\T\vec{X}
\end{bmatrix}
+
\begin{bmatrix}
    0 & \vec{0}_{1,d-1} \\ 
    \vec{0}_{d-1,1} & \vec{H}^\T\vec{H}
\end{bmatrix}.
\]
\end{lemma}

\begin{proof}
Solely based on $\vec{x}$, extend to an orthonormal matrix:
\[
\hat{\vec{X}} 
= \begin{bmatrix}
    \vec{x} & \vec{X}
\end{bmatrix}.
\]
For instance, append the identity to $\vec{x}$, delete the first column which is dependent on the previous columns, then orthonormalize sequentially using Gram--Schmidt.

Since $\vec{y} = \vec{G}^\T\vec{G}\vec{x}$,  
\[
\hat{\vec{X}}^\T \vec{G}^\T\vec{G} \hat{\vec{X}}
= 
\begin{bmatrix}
    \vec{x}^\T \vec{y} & \vec{y}^\T\vec{X} \\
    \vec{X}^\T \vec{y} & \vec{X}^\T\vec{G}^\T\vec{G}\vec{X}
\end{bmatrix}.
\]
The first row and column of this matrix depend only on $\vec{x}$ and $\vec{y}$.
We will now show $\vec{X}^\T\vec{G}^\T\vec{G}\vec{X}$ is the sum of a matrix depending on $\vec{x}$ and $\vec{y}$ and Gaussian matrix independent of $\vec{x}$ and $\vec{y}$.

Define $\vec{r} = \|\vec{G}\vec{x}\|^{-1}\vec{G}\vec{x}$, and note that $\|\vec{r}\| = 1$.
Solely based on $\vec{r}$, extend to an orthonormal matrix:
\[
\hat{\vec{R}} = \begin{bmatrix}
    \vec{r} & \vec{R}
\end{bmatrix}.
\]
Since $\hat{\vec{R}}$ is orthonormal, $\vec{r}\vec{r}^\T + \vec{R}\vec{R}^\T = \vec{I}$, and
\begin{align*}
\vec{X}^\T\vec{G}^\T\vec{G}\vec{X}
&= 
\vec{X}^\T\vec{G}^\T\vec{r}\vec{r}^\T\vec{G}\vec{X}
+ \vec{X}^\T\vec{G}^\T\vec{R}\vec{R}^\T\vec{G}\vec{X}
\\&=
\|\vec{G}\vec{x}\|^{-2} \vec{X}^\T\vec{G}^\T\vec{G}\vec{x}\vec{x}^\T\vec{G}^\T\vec{G}\vec{X}
+ \vec{X}^\T\vec{G}^\T\vec{R}\vec{R}^\T\vec{G}\vec{X}.
\end{align*}
Thus, using that $\|\vec{G}\vec{x}\|^2 = \vec{x}^\T\vec{G}^\T\vec{G}\vec{x} = \vec{x}^\T\vec{y}$, 
\[
\hat{\vec{X}}^\T \vec{G}^\T\vec{G}\hat{\vec{X}}  
= 
\begin{bmatrix}
    \vec{x}^\T \vec{y} & \vec{y}^\T\vec{X} \\
    \vec{X}^\T \vec{y} & (\vec{x}^\T\vec{y})^{-1}\vec{X}^\T\vec{y}\vec{y}^\T\vec{X}
\end{bmatrix}
+
\begin{bmatrix}
    0 & \vec{0}_{1,n-1} \\ 
    \vec{0}_{n-1,1} & \vec{X}^\T\vec{G}^\T\vec{R}\vec{R}^\T\vec{G}\vec{X}
\end{bmatrix}.
\]

It remains to show $\vec{R}^\T\vec{G}\vec{X}$ is a $(r-1)\times(d-1)$ Gaussian matrix independent of $\vec{x}$ and $\vec{y}$.

First, note that since ${\vec{X}}$ is chosen only based on $\vec{x}$ (which is independent of $\vec{G}$), 
$\vec{G}\hat{\vec{X}}$ consists of iid Gaussians. 
Thus, the columns of $\vec{G}{\vec{X}}$ are mutually independent of one another and $\vec{x}$, and hence $\vec{G}\vec{X}$ is mutually independent of $\vec{x}$ and $\vec{y}$.
Finally, since $\vec{R}$ depends only on $\vec{r} = (\vec{x}^\T\vec{y})^{-1/2}\vec{G}\vec{x}$, $\vec{R}$ is independent of $\vec{G}\vec{X}$.
Thus, $\vec{R}^\T\vec{G}\vec{X}$ has iid Gaussian entries independent of $\vec{x}$ and $\vec{y}$ (and hence any matrices constructed solely from $\vec{x}$ and $\vec{y}$).
\end{proof}

We will now prove the general statement. 

\begin{proof}[Proof of \cref{thm:wishart}]
We proceed by induction. 
Suppose that after $t$ queries, the result of the lemma holds.
Let $\vec{x}_{t+1}$ be a query chosen based solely on $\vec{x}_1, \ldots, \vec{x}_t$ and $\vec{y}_1, \ldots, \vec{y}_t$ and hence independent of $\vec{G}_t$.

Then by the inductive hypothesis,
\[
    \vec{G}^\T\vec{G}\vec{x}_{t+1} = 
    \vec{V}_t \vec{\Delta}_t\vec{V}_t^\T\vec{x}_{t+1}  +  \vec{V}_t^\T\begin{bmatrix}
        \vec{0}_{t,t} & \vec{0}_{t,d-t} \\
        \vec{0}_{d-t,t} & \vec{G}_t^\T\vec{G}_t
    \end{bmatrix}\vec{V}_t^\T\vec{x}_{t+1}.
\]
Let $\vec{x}$ denote the bottom $d-t$ entries of $\vec{V}_t^\T\vec{x}_{t+1}$, normalized to have length 1. 
By \cref{thm:wishart_1step},
querying $\vec{G}_t^\T \vec{G}_t$ results in the factorization
\[
\hat{\vec{X}}^\T \vec{G}_t^\T \vec{G}_t\hat{\vec{X}}= \vec{\Delta} + 
\begin{bmatrix}
    0 & \vec{0}_{1,d-t-1} \\ 
    \vec{0}_{d-t-1,1} &  \vec{H}^\T\vec{H}
\end{bmatrix},
\]
where $\vec{\Delta}$ and $\hat{\vec{X}}$ are constructed solely as functions of $\vec{x}$ and $\vec{G}_t^\T \vec{G}_t \vec{x}$ (and hence of $\vec{x}_1, \ldots, \vec{x}_{t+1}$ and $\vec{y}_1, \ldots, \vec{y}_{t+1}$), $\hat{\vec{X}}$ is orthonormal, and $\vec{H}\sim\operatorname{Gaussian}(d-t-1,r-t-1)$ is independent of $\vec{x}$ (and hence of $\vec{x}_1, \ldots, \vec{x}_{t+1}$ and $\vec{y}_1, \ldots, \vec{y}_{t+1}$).

Define $\vec{G}_{t+1} = \vec{H}$ and the matrices
\[
\widetilde{\vec{X}}=
\begin{bmatrix}
    1 & \vec{0}_{1,t-1} \\ 
    \vec{0}_{t-1,t} & \hat{\vec{X}}
\end{bmatrix},
\qquad
\vec{V}_{t+1} 
= \vec{V}_t \widetilde{\vec{X}}
,\quad
\vec{\Delta}_{t+1} = \widetilde{\vec{X}}^\T \vec{\Delta}_t \widetilde{\vec{X}}
+ \begin{bmatrix}
    \vec{0}_{t,t} & \vec{0}_{t,d-t} \\
    \vec{0}_{d-t,t} & \vec{\Delta}
\end{bmatrix}
.
\]
Clearly $\vec{V}_{t+1}$ is orthonormal and $\vec{V}_{t+1}$ and $\vec{\Delta}_{t+1}$ are constructed solely as functions of $\vec{x}_1, \ldots, \vec{x}_{t+1}$ and $\vec{y}_1, \ldots, \vec{y}_{t+1}$
We easily verify that
\[
\vec{V}_{t+1}^\T\vec{G}^\T\vec{G}\vec{V}_{t+1}
=
\widetilde{\vec{X}}^\T \vec{\Delta}_t \widetilde{\vec{X}}
+
\begin{bmatrix}
    \vec{0}_{t,t} & \vec{0}_{t,d-t} \\
    \vec{0}_{d-t,t} & \hat{\vec{X}}^\T \vec{G}_t^\T\vec{G}_t \hat{\vec{X}}
\end{bmatrix}
=
\vec{\Delta}_{t+1}
+
\begin{bmatrix}
    \vec{0}_{t+1,t+1} & \vec{0}_{t+1,d-(t+1)} \\
    \vec{0}_{d-(t+1),t+1} & \vec{G}_{t+1}^\T \vec{G}_{t+1}
\end{bmatrix}
\]

The result is proved as the base case $t=0$ is trivial.
\end{proof}

\subsection{Anti-concentration for sums}
\label{sec:anticoncentration-sums}

In this section we will prove \cref{thm:sum-anticoncentration,thm:gausian-inner-prod}.

We begin recalling the Berry--Esseen Theorem for non-identically distributed summands.

\begin{proposition}[{\protect Berry--Esseen; see e.g. \cite[\S 1.9]{serfling_80}}]\label{thm:berry_esseen}
    There exists a constant $C>0$ such that, if $X_1, \ldots, X_k$ are independent random variables with $\EE[X_i] = 0$, $\EE[X_i^2] = \sigma_i^2$, and $\EE[|X_i|^3] = \rho_i$, and if we define 
    \[
    Y = \frac{X_1 + \cdots + X_k}{(\sigma_1^2 + \cdots + \sigma_k^2)^{1/2}},
    \]
    then the CDF $F_Y$ of $Y$ is near to the CDF $\Phi$ of a standard Gaussian in that,
    \[
        \forall z: |F_Y(z) - \Phi(z)| \leq C\frac{\rho_1 + \cdots + \rho_k}{(\sigma_1^2 + \cdots + \sigma_k^2)^{3/2}}.
    \]
\end{proposition}

In addition, we note a simple anti-concentration bound for Gaussians.
\begin{lemma}\label{thm:gaussian_anticoncentration}
    Let $Z\sim\mathcal{N}(0,1)$. Then, for any $t\in \R$ and any $\alpha > 0$,
    \[
    \PP\Bigl[ |Z-t|  <  \alpha \Bigr] < \sqrt{\frac{2}{\pi}}\,\alpha.
    \]
\end{lemma}
\begin{proof}
Note that the density for $Z$ is $f_Z(x) = (1/\sqrt{2\pi}) \exp(-x^2/2)$, and $f_Z(x) \leq 1/\sqrt{2\pi}$ for all $x$.
Hence, 
\[
\PP\bigl[ |Z-t| < \alpha \bigr]
= \int_{t-\alpha}^{t+\alpha} f_Z(x) \d{x}
\leq \frac{2\alpha}{\sqrt{2\pi}}
= \sqrt{\frac{2}{\pi}} \alpha.
\qedhere
\]
\end{proof}

\begin{proof}[Proof of \cref{thm:sum-anticoncentration}]
Without loss of generality, we can assume $\EE[X_i] = 0$ for all $i$ by absorbing the means into $t$.
By assumption we have that
\[
\hat{\sigma} = \sqrt{\sigma_1^2 + \cdots + \sigma_k^2} \geq \sigma \sqrt{k}
,\qquad 
\rho_1 + \cdots + \rho_k \leq k \rho.
\]
Let $Y = X / \hat{\sigma}$ with cumulative distribution $F_Y$ and let $t' = t / \hat{\sigma}$. 
Then, with $C$ denoting the constant from \cref{thm:berry_esseen}, we apply \cref{thm:berry_esseen}, \cref{thm:gaussian_anticoncentration}, and the bounds above to obtain
\begin{align*}
\PP\bigl[ |X-t| < \alpha \sigma \sqrt{k} \bigr]
& \leq \PP\bigl[ |X-t| < \alpha \hat{\sigma} \bigr]
\\&=\PP\bigl[ |Y-t'| < \alpha \bigr]
\\&= F_Y(t'+\alpha ) - F_Y(t'-\alpha)
\\& \leq \Phi(t'+\alpha) - \Phi(t'-\alpha) + \frac{2Ck\rho}{k^{3/2}\sigma^3}
\\&  \leq \sqrt{\frac{2}{\pi}} \alpha + \frac{2C\rho}{\sigma^3\sqrt{k}}.
\end{align*}
Since $(1-\sqrt{2/\pi}) > 0.2$, the result follows by the choice $k > 100 C^2\rho^2 / (\alpha^2 \sigma^6)$, and relabeling $C$.
\end{proof}

\begin{proof}[Proof of \cref{thm:gausian-inner-prod}]
Let $\vec {g}_\ell$ denote the $\ell$-th row of $\vec G$ and define $x_\ell = \vec{u}^\T\vec{g}_\ell$ and $y_\ell = \vec{v}^\T \vec{g}_\ell$.
In order to apply \cref{thm:sum-anticoncentration} to the sum
\begin{align*}
    \vec{x}^\T \vec{y} = \sum_{\ell=1}^k x_\ell y_\ell,
\end{align*}
 which has independent terms since the $\vec{g}_\ell$ are independent, we must obtain upper-and lower-bounds on the variance and an upper bound on the third centered absolute moment of each term.

We will first bound the variance.
By direct computation (see \cref{thm:gaussianqf} for a derivation),
\[
    \VV\bigl[x_\ell  y_\ell\bigr] 
    = \VV\bigl[\vec g_\ell^\T \vec{v}\vec{u}^\T \vec g_\ell \bigr]
    = \| \vec{v}\vec{u}^\T + \vec{u}\vec{v}^\T \|_\F^2/2
    \geq  \| \vec{u} \|_2^2 \|\vec{v}\|_2^2 .
\]
Here we have used that 
\[
\| \vec{v}\vec{u}^\T + \vec{u} \vec{v}^\T \|_\F^2
= 2 \|\vec{v}\vec{u}^\T \|_\F^2 + 2(\vec{v}^\T\vec{u})^2
\geq 2\|\vec{v}\|_2^2 \|\vec{u}\|_2^2.
\]
We now argue the third absolute moments are bounded. 
Note that  $x_\ell$ and $y_\ell$ are both normally distributed with mean zero and variance at most $1$ (since $\| \vec{u} \|_2 \le 1$ and $\| \vec{v} \|_2 \le 1$).
Then $x_{\ell}y_{\ell}$ is sub-exponential with constant width parameter \cite[Lemma 2.7.7]{vershynin_18}, and hence $\EE[ | x_\ell y_\ell - \EE[x_\ell y_\ell] |^3 ] \leq \rho$ for some $\rho$.

Hence, applying \cref{thm:sum-anticoncentration}, for any $\alpha > 0$, and provided $k > C \rho^2 / (\alpha^2 \| \vec{u} \|_2^6 \|\vec{v}\|_2^6)$, where $C$ is the constant from \cref{thm:sum-anticoncentration},
\begin{equation*}
    \PP\Bigl[  \big| \vec{x}^\T\vec{y} - t \big|  < \alpha\|\vec{u}\|_2 \| \vec{v} \|_2 \sqrt{k} \Bigr] 
    < \alpha.
\end{equation*}
Relabeling $C$ gives the result.
\end{proof}

\subsection{Other facts}

\begin{fact}\label{thm:gaussianqf}
For a matrix $\vec{A}\in\R^{d\times d}$, if $\vec{g}\sim\operatorname{Gaussian}(d,1)$, then $\VV[\vec{g}^\T\vec{A}\vec{g}] = \|\vec{A}+\vec{A}^\T\|_\F^2/2$.
\end{fact}

\begin{proof}
    Since $\vec{g}^\T \vec{A}\vec{g} = \vec{g}^\T (\vec{A}^\T + \vec{A}) \vec{g}/2$, without loss of generality we can assume $\vec{A}$ is symmetric.
    Let $\vec{A} = \vec{U}\vec{\Lambda}\vec{U}^\T$ be the eigendecomposition of $\vec{A}$, then $\vec{h} = \vec{U}^\T\vec{g}\sim \operatorname{Gaussian}(d,1)$ and $\vec{h}^\T\vec{\Lambda}\vec{h}$ is a linear-combination of independent Chi-squared random variables with one degree of freedom (and variance 2). The result follows since $\|\vec{A}\|_\F^2 = \|\vec{\Lambda}\|_\F^2$.
\end{proof}

\begin{fact}\label{thm:Aexpectednorm}
For $r,d\geq 1$, suppose $\vec{G}\sim\operatorname{Gaussian}(r,d)$. 
Then 
\[  
    \EE\big[ \| \vec{I} \circ \vec{G}^\T\vec{G}  \|_\F^2 \big]  =d(2r+r^2),
    \qquad
    \EE\big[ \| \vec{G}^\T\vec{G} - \vec{I} \circ \vec{G}^\T\vec{G}  \|_\F^2 \big] = (d^2-d)r.
\]
\end{fact}

\begin{proof}
    Write $\vec{A} = \vec{G}^\T\vec{G}$.
    The diagonal entries of $\vec{A}$ are distributed as Chi-squared random variables with $r$ degrees of freedom. These entries have mean $r$ and variance $2r$. 
    Likewise, the off-diagonal entries of $\vec{A}$ are distributed as the inner product of two independent standard normal Gaussian vectors of length $r$. 
    These entries therefore have mean zero and variance $r$.
    Therefore, for $i\neq j$, 
    \[
    \EE\bigl [ [\vec{A}]_{i,i}^2 \bigr] 
    =\VV\bigl [ [\vec{A}]_{i,i} \bigr]  + \EE\bigl [ [\vec{A}]_{i,i} \bigr]^2
    = 2r+r^2
    ,\qquad
    \EE\bigl [ [\vec{A}]_{i,j}^2 \bigr] =
    \VV\bigl [ [\vec{A}]_{i,j} \bigr] 
    = r.
    \]
    The result follows by linearity of expectation.
\end{proof}

\section{High probability algorithm}\label{sec:high-prob}

The bound \cref{thm:ub_main} for \cref{alg:main} has an unfavorable dependence $O(1/\delta)$ on the failure probability $\delta$. 
We will now use a high-dimensional version of the ``median trick'' to improve the dependence on the failure probability to logarithmic.

\begin{algorithm}[ht]
\caption{Fixed-sparse-matrix recovery (boosted)}\label{alg:main-boosted}
\fontsize{10}{14}\selectfont
\begin{algorithmic}[1]
\Procedure{boosted-fixed-sparse-matrix recovery}{$\vec{A},\vec{S},m,r$}
\State Run \cref{alg:main} independently $r$ times to get $\widetilde{\vec{A}}_1, \ldots, \widetilde{\vec{A}}_r$
\State $\forall i,j$: define $d_{i,j} = \| \widetilde{\vec{A}}_i - \widetilde{\vec{A}}_j \|_\F$ 
\State $\forall i$: define $B_i$ as the $\lceil r/2\rceil$-th smallest $d_{i,j}$ \Comment{$\big|\{ j\in[r] : d_{i,j} \leq B_i \}\big| = \lceil r/2\rceil$}
\State Compute $i^* = \operatornamewithlimits{argmin}_{i} B_i$ \Comment{$\forall i: B_{i^*} \leq B_{i}$}
\EndProcedure
\State \Return $\widetilde{\vec{A}}_{i^*}$ 
\end{algorithmic}
\end{algorithm}

\begin{theorem}
    Consider any $\vec{A} \in\mathbb{R}^{n\times d}$ and any $\vec{S}\in\{0,1\}^{n\times d}$ with at most $s$ nonzero entries per row. 
    For any $\varepsilon>0$ and $\delta \in (0,1)$, if $m\geq s+2$ and additionally
    \[
    m\geq s \left(\frac{90}{\varepsilon} + 1 \right) + 1 ,
    \quad
    r \geq 10\log\left(\frac{1}{\delta}\right),
    \]
    then, using $m\cdot r$ matrix-vector queries, \cref{alg:main-boosted} returns a matrix $\widetilde{\vec{A}}$ satisfying:
    \[
    \PP\bigl[ \| \vec{A} - \widetilde{\vec{A}} \|_{\F} < (1+\varepsilon) \| \vec{A} - \vec{S}\circ \vec{A} \|_\F \bigr] 
    \geq 1 - \delta.
    \]
\end{theorem}

\begin{proof}
Let $\tilde \varepsilon = \frac29 \varepsilon$.
Define the set 
\[
P = \big \{ i \in [r] : \|\vec{S}\circ\vec{A} - \widetilde{\vec{A}}_i\|_\F^2 \leq \tilde\varepsilon \| \vec{A} - \vec{S}\circ \vec{A} \|_\F^2 \big \}.
\]
In \cref{eqn:ub_main:SAA} of the proof of \cref{thm:ub_main_prob}, we show that
\[
\PP\bigl[\|\vec{S}\circ\vec{A} - \widetilde{\vec{A}}_i\|_\F^2 \geq \varepsilon \|\vec{A} - \vec{S}\circ \vec{A}\  |_\F^2 \bigr]
\leq s/\big((m-s-1)\tilde\varepsilon\big).
\]
Hence, if $m\geq s((20/\tilde\varepsilon)+1) + 1$, then $\PP[i\in P] \geq \frac{19}{20}$.
Define the event
\[
E = \{ |P| > \lfloor r/2 \rfloor \}.
\]
A standard result \cite[Prop 2.4 (a)]{angluin_valiant_79} asserts that with $q=19/20$, 
\[
\sum_{k=0}^{\lfloor r/2\rfloor} \binom{r}{k} q^k (1-q)^{r-k}
\leq \exp\left( - \frac{rq}{2}\left(1-\frac{1}{2q}\right)^2\right)
= \exp\left(- \frac{81r}{760} \right).
\]
In addition, $\delta \geq \exp(-r/10)$ by definition of $r$.
Therefore, $\PP[E] \geq 1- \delta$.

We will condition on $E$ for the remainder of the proof.
By the triangle inequality, for any indices $i,j\in P$,
\[
d_{i,j} = \| \widetilde{\vec{A}}_i - \widetilde{\vec{A}}_j \|_\F
\leq \| \widetilde{\vec{A}}_i - \vec{S}\circ\vec{A} \|_\F + \|\vec{S}\circ\vec{A} - \widetilde{\vec{A}}_j \|_\F
\leq 2\sqrt{\varepsilon} \| \vec{A} - \vec{S}\circ \vec{A} \|_\F.
\]
Since $|P| > \lfloor r/2 \rfloor$, then for each $i\in P$, there are at least $\lfloor r/2 \rfloor$ indices $j$ satisfying $d_{i,j} \leq 2\sqrt{\varepsilon} \| \vec{A} - \vec{S}\circ \vec{A} \|_\F$. Thus by definition of $B_i$, 
\[
B_i \leq 2 \sqrt{\tilde\varepsilon} \| \vec{A} - \vec{S}\circ \vec{A} \|_\F.
\]
By definition of $B_{i^*}$, there are at least $\lceil r/2 \rceil$ indices $j$ for which
\[ 
\| \widetilde{\vec{A}}_{i^*} - \widetilde{\vec{A}}_{j} \|_\F \leq 2 \sqrt{\tilde\varepsilon} \| \vec{A} - \vec{S}\circ \vec{A} \|_\F
\]
Simultaneously, $|P|>\lfloor r/2 \rfloor$. 
Since $|P| + \lceil r/2\rceil > \lfloor r/2 \rfloor + \lceil r/2\rceil = r$, the pigeonhole principle ensures there is at least one $j^*$ for which $j^*\in P$ and
\[
\| \widetilde{\vec{A}}_{i^*} - \widetilde{\vec{A}}_{j^*} \|_\F \leq 2 \sqrt{\tilde\varepsilon} \| \vec{A} - \vec{S}\circ \vec{A} \|_\F.
\]
Applying the triangle inequality, we find that 
\[
\| \vec{S}\circ \vec{A} - \widetilde{\vec{A}}_{i^*} \|_\F 
\leq \| \vec{S}\circ \vec{A} - \widetilde{\vec{A}}_{j^*} \|_\F + \| \widetilde{\vec{A}}_{i^*} - \widetilde{\vec{A}}_{j^*} \|_\F
\leq 3 \sqrt{\tilde\varepsilon} \| \vec{A} - \vec{S}\circ \vec{A} \|_\F.
\]
which implies $ \|\vec{S}\circ \vec{A} - \widetilde{\vec{A}}_{i^*} \|_\F^2 \leq 9 \tilde\varepsilon \| \vec{A} - \vec{S}\circ \vec{A} \|_\F^2.$
Adding $\| \vec{A} - \vec{S}\circ \vec{A} \|_\F^2$ to both sides (see \cref{eqn:AAtilde-decomp}) and taking square roots,
\[ \|\vec{A} - \widetilde{\vec{A}}_{i^*} \|_\F \leq \sqrt{1+9 \tilde\varepsilon} \| \vec{A} - \vec{S}\circ \vec{A} \|_\F = \sqrt{1+2\varepsilon} \| \vec{A} - \vec{S}\circ \vec{A} \|_\F \leq (1 + \varepsilon) \| \vec{A} - \vec{S}\circ \vec{A} \|_\F.\qedhere\]
\end{proof}

\printbibliography[]

\end{document}